\documentclass[conference,10pt]{IEEEtran}
\IEEEoverridecommandlockouts
\ifCLASSINFOpdf
\else
\fi

\usepackage{amssymb}
\usepackage{amsthm}
\usepackage{framed}
\usepackage{graphicx}
\usepackage{amsmath}
\usepackage{mathtools,amssymb}
\usepackage{nicefrac}
\usepackage{graphicx}
\usepackage{array}
\usepackage{stmaryrd}
\usepackage[textsize=small,disable]{todonotes}
\usepackage{xspace}
\usepackage{infer}
\usepackage{needspace}
\usepackage{marginnote}
\usepackage{tablefootnote}
\usepackage{xcolor}
\usepackage[final,colorlinks=true,linktocpage=true,hyperfootnotes]{hyperref}
\usepackage{tensor}
\usepackage{tikz}
\usetikzlibrary{arrows,decorations.pathmorphing,positioning,fit,trees,shapes,shadows,automata,calc}
\usepackage[normalem]{ulem}

\allowdisplaybreaks

\newcommand{\ie}{i.e.\@\xspace}

\newcommand{\eg}{e.g.\@\xspace}

\DeclareMathAlphabet{\mathpzc}{OT1}{pzc}{m}{it}

\newcommand{\pProgs}{\textnormal{\textsf{pProgs}}\xspace}   
\newcommand{\Vars}{\ensuremath{\mathsf{Var}}\xspace}   
\newcommand{\Vals}{\ensuremath{\mathsf{Val}}\xspace}    
\newcommand{\DExprs}{\ensuremath{\mathsf{PGuards}}\xspace}    

\newcommand{\appProgs}{\Stmt}

\newcommand{\E}{\ensuremath{\mathbb{E}_{\geq 0}^{{\infty}}}}
\newcommand{\Epm}{\ensuremath{\mathbb{E}^{{\star}}}}
\newcommand{\Pairs}{\ensuremath{\mathbb{P}}}
\newcommand{\EqPairs}{\ensuremath{\mathbb{I\hspace{-.1ex}E}}}

\newcommand{\States}{\ensuremath{\Sigma}}
\newcommand{\R}{\ensuremath{\mathbb{R}}}
\newcommand{\Rpos}{\ensuremath{\R_{{}\geq 0}}}
\newcommand{\Rposinf}{\ensuremath{\Rpos^{\infty}}}

\newcommand{\SKIP}{\ensuremath{\textnormal{\texttt{skip}}}}

\newcommand{\AssignSymbol}{\mathrel{\textnormal{\texttt{:=}}}}
\newcommand{\ASSIGN}[2]{\ensuremath{#1 \AssignSymbol #2}}

\newcommand{\COMPOSE}[2]{\ensuremath{#1\textnormal{\texttt{;}}\: #2}}

\newcommand{\ITE}[3]{\ensuremath{\textnormal{\texttt{if}}\left(#1\right)\left\{#2\right\}\textnormal{\texttt{else}}\left\{#3\right\}}}
\newcommand{\If}{\ensuremath{\textnormal{\texttt{if}}}}
\newcommand{\Else}{\ensuremath{\textnormal{\texttt{else}}}}

\newcommand{\WHILEDO}[2]{\ensuremath{\textnormal{\texttt{while}}\left(#1\right)\left\{#2\right\}}}
\newcommand{\WHILE}{\ensuremath{\textnormal{\texttt{while}}}}

\newcommand{\wpsymbol}{\ensuremath{\textnormal{\textsf{wp}}}}
\newcommand{\boldwpsymbol}{\ensuremath{\textnormal{\textsf{\textbf{wp}}}}}

\renewcommand{\wp}[2]{\ensuremath{\wpsymbol\left[{#1}\right]\left({#2}\right)}}

\newcommand{\wpeqpsymbol}{\ensuremath{\widetilde{\textnormal{\textsf{wp}}}}}

\newcommand{\wpeqp}[3]{\ensuremath{\wpeqpsymbol\left[{#1}\right]\EqPair{#2}{#3}}}

\newcommand{\bigLbag}{\raisebox{-.2em}{\text{\LARGE $\Lbag$}}}
\newcommand{\bigRbag}{\raisebox{-.2em}{\text{\text{\LARGE $\Rbag$}}}}
\newcommand{\hugeLbag}{\raisebox{-.2em}{\text{\Huge $\Lbag$}}}
\newcommand{\hugeRbag}{\raisebox{-.2em}{\text{\text{\Huge $\Rbag$}}}}

\newcommand{\charwp}[3]{\tensor*[^{#1}_{#2}]{F}{_{{#3}}}}
\newcommand{\charwpn}[4]{\tensor*[^{#1}_{#2}]{F}{_{{#3}}^{#4}}}

\newcommand{\positive}[1]{\tensor*[^{+}]{#1}{}}
\newcommand{\negative}[1]{\tensor*[^{-}]{#1}{}}
\newcommand{\Exp}[2]{\ensuremath{\textnormal{\textsf{E}}_{#1}\left({#2}\right)}}
\newcommand{\subst}[2]{\ensuremath{\left[{#1}/{#2}\right]}}

\newcommand{\zero}{\ensuremath{\boldsymbol{0}}}

\newcommand{\lfp}{\ensuremath{\textnormal{\textsf{lfp}}\,}}

\newcommand{\eval}[1]{\ensuremath{\llbracket {#1} \rrbracket}}
\newcommand{\probof}[2]{\ensuremath{\llbracket {#1} \colon {#2} \rrbracket}}

\newcommand{\ctert}[1]{\ensuremath{\mathbf{#1}}}

\newcommand{\ExpToFun}[1]{\eval{#1}}

\newcommand{\true}{\ensuremath{\mathsf{true}}\xspace}
\newcommand{\false}{\ensuremath{\mathsf{false}}\xspace}

\newcommand{\To}{\rightarrow}                          

\newcommand{\sem}[1]{\llbracket #1 \rrbracket}

\newcommand{\Stmt}{\ensuremath{\textnormal{\textsf{pProgs}}}\xspace}

\newcommand{\stmt}{\ensuremath{C}\xspace}

\newcommand{\pguard}{\ensuremath{\xi}\xspace}

\newcommand{\mydot}{\text{\LARGE\textbf{.}~}}

\newcommand{\Pair}[2]{({#1},\, {#2})}
\newcommand{\EqPair}[2]{\Lbag{#1},\, {#2}\Rbag}
\newcommand{\bigEqPair}[2]{\bigLbag{#1},~\, {#2}\bigRbag}
\newcommand{\hugeEqPair}[2]{\hugeLbag{#1},~\, {#2}\hugeRbag}

\theoremstyle{plain}
\newtheorem{theorem}{Theorem}
\newtheorem{lemma}{Lemma}

\theoremstyle{definition}
\newtheorem{example}{Example}
\newtheorem{definition}{Definition}

\theoremstyle{remark}
\newtheorem{remark}{Remark}

\begin{document}
%
\title{A Weakest Pre--Expectation Semantics for Mixed--Sign Expectations\thanks{This work was supported by the Excellence Initiative of the German federal and state government and by the CDZ project CAP (GZ 1023).}}



%
\author{\IEEEauthorblockN{Benjamin Lucien Kaminski\IEEEauthorrefmark{2}\IEEEauthorrefmark{3} and
Joost-Pieter Katoen\IEEEauthorrefmark{2}}
\IEEEauthorblockA{\IEEEauthorrefmark{2}Software Modeling and Verification Group\\
RWTH Aachen University, Germany\\
Email: \{benjamin.kaminski,katoen\}@cs.rwth-aachen.de}
\IEEEauthorblockA{\IEEEauthorrefmark{3}Currently on leave at the Programming Principles, Logic and Verification Group\\
University College London, United Kingdom}}


\maketitle

\begin{abstract}
We present a weakest--precondition--style calculus for reasoning about the expected values (pre--expectations) of \emph{mixed--sign unbounded} random variables after execution of a probabilistic program. 
The semantics of a while--loop is defined as the limit of iteratively applying a functional to a zero--element just as in the traditional weakest pre--expectation calculus, even though a standard least fixed point argument is not applicable in our semantics.
A striking feature of our semantics is that it is always well--defined, even if the expected values do not exist.
We show that the calculus is sound and allows for compositional reasoning.
Furthermore, we present an invariant--based approach for reasoning about pre--expectations of loops.
\end{abstract}


%
\IEEEpeerreviewmaketitle

\section{Introduction}
\label{sec:intro}
\noindent
Probabilistic programs are programs that support choi{\-}ces like ``execute program $C_1$ with probability $\nicefrac{1}{3}$ and program $C_2$ with probability $\nicefrac{2}{3}$".
Describing randomized algorithms has been the classical application of probabilistic programs.
Applications in biology, machine learning, quantum computing, security, and so on, have recently led to a rapidly growing interest in such programs~\cite{DBLP:conf/icse/GordonHNR14}.
Although probabilistic programs syntactically are normal--looking programs, reasoning about their correctness is intricate.
The key property of program termination exemplifies this.
Whereas a classical program terminates or not, this is no longer true for probabilistic programs.
They can diverge, but this may happen with probability 0.
In addition, in contrast to classical programs that either do not terminate at all or terminate in finitely many steps, a probabilistic program may take infinitely many steps on average to terminate, even if its termination probability is 1.

Establishing correctness of probabilistic programs needs---even more so than ordinary programs---formal reasoning.
Weakest--precondition ($\wpsymbol$) calculi \`a la Dijk{\-}stra~\cite{Dijkstra} provide an important tool to enable formal reasoning.
To develop such calculi for probabilistic programs, one has to take into account that due to its random nature, the final state of a program on termination need not be unique.
Thus, rather than a mapping from inputs to outputs (as in Dijkstra's approach), probabilistic programs can be thought of mapping an initial state to a \emph{distribution} over possible final states.
More precisely, we may obtain \emph{sub--}distributions where the ``missing'' probability mass represents the likelihood of divergence.
Given a random variable $f$ (e.g.\ $f = x^2 + y$, where $x$ and $y$ are program variables) and an initial state $\sigma$, a key issue is to determine $f$'s expected value\footnote{Commonly called pre--expectation~\cite{mciver}.} on the probabilistic program's termination.
This was first studied in Kozen's seminal work on probabilistic propositional dynamic logic (PPDL)~\cite{DBLP:journals/jcss/Kozen85}.
Its box-- and diamond--modalities provide probabilistic versions of Dijkstra's weakest (liberal) preconditions.
Amongst others, Jones~\cite{jones}, Hehner~\cite{Hehner:FAC:2011}, and McIver~\&~Morgan~\cite{mciver} have furthered this line of research, e.g.\ by considering non--determinism and proof rules for loops.
Recently, Kaminski \emph{et al.}~\cite{DBLP:conf/esop/KaminskiKMO16} provided $\wpsymbol$--style reasoning about the expected run--time of probabilistic programs while Olmedo \emph{et al.}~\cite{DBLP:journals/corr/OlmedoKKM16} consider recursion.

All these works (except PPDL) make an important---though res{\-}tric{\-}tive---as{\-}sump{\-}tion: the random variable $f$ maps program states to the \emph{non--negative} reals.
In McIver \& Morgan's terminology, such random variables $f$ are called \emph{expectations}.
That is to say, the aforementioned approaches do not deal with \emph{mixed--sign} expectations, i.e.\ expectations that can be negative, or even negative and positive.
\mbox{McIver \& Morgan~\cite[pp.\ 70]{mciver}} forbid mixed--sign expectations altogether and argue that ``For mixed--sign or unbounded expectations, however, well--definedness is not assured: such cases must be treated individually. [\ldots] That is, although [a program] itself may be well defined, the greatest pre--expectation [for $f = (-2)^n$] is not---and that is a good reason for avoiding mixed signs in general.''
A workaround is to assume bounded negative values~\cite{DBLP:journals/tcs/McIverM01a}, but this also provides no general solution.

An exception to the widespread and generally condoned neglect of unbounded mixed--sign expectations is Kozen's PPDL~\cite{DBLP:journals/jcss/Kozen85} as it provides an expectation transformer semantics for probabilistic programs with respect to general measurable post--expectations $f$ and thus does \emph{not} forbid mixed--sign expectations altogether.
PPDL's proof rule for \emph{reasoning} about while loops, however, requires $f$ to be non--negative~\mbox{\cite[Section 4, page 168: the ``while rule"]{DBLP:journals/jcss/Kozen85}}.
This proof rule is hence unfit for reasoning about mixed--sign expectations.
In fact, three out of four rules of the deduction system of PPDL that deal with iteration (and therefore with loops) require $f$ to be non--negative and are hence not applicable to reasoning about mixed--sign post--expectations $f$~\cite[Section 4: Rules (8), (9), and the ``while rule"]{DBLP:journals/jcss/Kozen85}.
The only exception to this is a rule that allows for upper bounding the \emph{pre--expectation} by a \emph{non--negative function}, even if $f$ is mixed--sign~\cite[Section 4: Rule (10)]{DBLP:journals/jcss/Kozen85}.
This rule, however, is insufficient for upper--bounding the pre--expectation by a negative value, which in practice can be desirable and is possible in our calculus, see Example \ref{ex:amortized}.

Another drawback of PPDL is that reasoning even about simple programs and properties can become quite involved, requiring a fairly high degree of mathematical reasoning, i.e.\ to say that PPDL requires a lot of reasoning inside the program semantics while the approach of McIver \& Morgan and the approach we present in this paper constitutes more of a syntactic reasoning on the source code level.
For example,~\mbox{\cite[Section 7]{DBLP:journals/jcss/Kozen85}} gives a circa two--page proof sketch of the expected run--time of a ``simple random walk" carried out in PPDL.
It requires a fair amount of domain--specific knowledge about integers and combinatorics and is thus not easily amenable to automation.
A complete proof of the expected run--time in the $\wpsymbol$--calculus \`{a} la McIver~\&~Morgan requires only a fraction of the effort (see Appendix~\ref{app:runtimewpproof}).
Partial automations of $\wpsymbol$--style proofs in theorem provers such as Isabelle/HOL have been developed~\cite{Hurd:TPHOL:02,DBLP:journals/afp/Cock14}.
A partial automation of $\wpsymbol$--style proofs for expected run--times in the vein of~\cite{DBLP:conf/esop/KaminskiKMO16} has recently been carried out by H\"{o}lzl~\cite{DBLP:conf/itp/Holzl16}.
The $\wpsymbol$--style calculus for mixed--sign expectations we present here is closely related to the standard weakest pre--expectation calculus and so we believe that existing automation techniques are likely to carry over easily.

At first sight, avoiding mixed--sign expectations looks like a minor technical restriction.
In practice it is not:
For instance, program variables may become negative during program execution, having a negative impact of $f$'s value.
As another example, the efficiency of data structures such as randomized splay trees~\cite{DBLP:journals/ipl/AlbersK02} is typically carried out using amortized analysis.
Such analysis is similar to expected run--time analysis, but is concerned with the cost averaged over a sequence of operations.
In the accounting and potential method in amortized analysis, a decrease in potential (or credit) ``pays for'' particularly expensive operations whereas increases model cheap operations.
The amortized cost $f$ during the execution of a probabilistic program may thus become arbitrarily negative.
Finally, we mention that negative expectations or even negative probabilities have applications in quantum computing and finance.\footnote{See \url{https://en.wikipedia.org/wiki/Negative_probability} and the various references therein.}

Current $\wpsymbol$--approaches do not handle the aforementioned scenarios off--the--shelf.
A workaround is to perform a Jordan decomposition of $f$ by $f = \positive{f} - \negative{f}$, where $\positive{f}$ and $\negative{f}$ are both non--negative expectations, and analyze $\positive{f}$ and $\negative{f}$ individually using the standard $\wpsymbol$--calculus.
This, however, can easily become quite involved, for example when trying to reason about the expected value of $x$ after execution of 
\begin{align*}
	\WHILEDO{\nicefrac{1}{2}}{\ASSIGN{x}{-x-\textsf{sign}(x)}}~.
\end{align*}
In every iteration, a fair coin is flipped to decide whether to terminate the loop or execute its body followed by a recursive execution of the entire loop.
Intuitively, this program computes a variant of a geometric distribution on $x$ where the sign alternates with increasing absolute value of $x$.
The expected value of $x$ after execution of the above program is given by $\nicefrac{x}{3} - \nicefrac{\textsf{sign}(x)}{9}$.
A detailed comparison of tackling this analysis by the methods presented in this paper to a Jordan--decomposition--based approach is provided in Appendix~\ref{app:comparison}.

Despite the existence of a mathematical theory of signed random variables, there are good reasons why they are avoided in current $\wpsymbol$--approaches: the notion of expectation needs to be reconsidered, and a complete partial order on these adapted expectations---key to defining the semantics of loopy programs---is required.
It turns out that this is not trivial.
It is this challenge that this paper attempts to take up. 
We provide a sound semantics of probabilistic programs that directly manipulates mixed--sign expectations $f$.
In particular, \emph{our semantics is always defined regardless of whether classical pre--expectations}~\cite{DBLP:journals/jcss/Kozen85,mciver,Hehner:FAC:2011} \emph{exist or not}.
We start by redefining what an expectation that can be negative in fact is.
The crux of our approach is to keep track of the integrability of the mixed--sign random variable $f$ by accompanying $f$ with a non--negative (but possibly infinite) expectation $g$ that bounds $|f|$.
Notice that we \emph{do not} require $f$ to be integrable as we want our semantics to be well--defined regardless of whether $f$ is integrable or not.
Instead, our semantics internally keeps track of $f$'s integrability.
We obtain a partial order by considering the kernel of a quasi--order on pairs $(f,\, g)$.
Equivalence classes under this kernel constitute the counterpart of expectations for the setting with mixed--sign random variables. 
This provides the basis for providing a sound $\wpsymbol$--calculus for reasoning about probabilistic programs with mixed--sign expectations.
In our setting, providing a sound semantics for loops cannot be done in the standard way, as Kleene's fixed point theorem is not applicable.
We therefore provide a direct proof.
An important ingredient to make this work is proving the existence of unique limits of sequences of equivalence classes of pairs $(f,g)$.
Moreover, we prove monotonicity and soundness of our novel weakest pre--expectation transformer.
This all is accompanied by a proof rule for reasoning about loops.
Various examples show the applicability of our transformer.

\paragraph*{Organization of the paper} 
In Section 2, we present syntax and effects of the probabilistic programming language that we build upon.
In Section 3, we revisit the traditional $\wpsymbol$--calculus and investigate the problems that would occur when naively letting the calculus act on mixed--sign expectations.
In Section 4, we present a new notion of mixed--sign expectations called \emph{integrability--witnessing expectations}, which incorporate bookkeeping for the integrability of the expectations.
In Section 5, we present a $\wpsymbol$--calculus acting on integrability--witnessing expectations.
In Section 6, we show that our calculus is sound and allows for monotonic reasoning.
Furthermore, we present an invariant rule for reasoning about loops and show its applicability.
We conclude with Section 7.

\section{The Probabilistic Programming Language}
\label{sec:language}
\noindent
In this section we present the probabilistic programming language used throughout this paper.
To model probabilistic programs we employ a standard imperative language \`a la Dijkstra's Guarded Command Language~\cite{Dijkstra} with a probabilistic feature: 
we allow for the guards that guard if--then--else constructs and while--loops to be probabilistic.
As an example, we allow for a program like
\begin{align*}
	&\WHILE \: \big(\nicefrac{2}{3} \cdot \langle x \text{ even} \rangle + \nicefrac{1}{3} \cdot \langle x \text{ odd} \rangle\big) \: \{ \ASSIGN{x}{x+1}\}
\end{align*}
which uses a probabilistic loop guard to establish a variant of a geometric distribution on the program variable $x$.
With probability $\nicefrac 2 3$ the loop \emph{terminates} if $x$ is odd and with probability $\nicefrac 1 3$ the loop terminates if $x$ is even.

Formally, the set of programs in the \emph{probabilistic guarded command language}, denoted \pProgs, is given by the grammar
\begin{align*}
\stmt  \quad\longrightarrow\quad &\SKIP ~~|~~ \ASSIGN{x}{E} ~~|~~ \COMPOSE{\stmt}{\stmt}\\
& ~~|~~ \ITE{\pguard}{\stmt}{\stmt} ~~|~~ \WHILEDO{\pguard}{\stmt}~.
\end{align*}
Here $x$ is a \emph{program variable} in \Vars, $E$ an arithmetical expression over program variables and $\pguard$ a \emph{probabilistic guard} in $\DExprs$.

To describe the effect of the different language constructs we first present some preliminaries.
A \emph{program state} $\sigma$ is a mapping from a finite set of program variables $\Vars$ to a countable set of values $\Vals$.  
Let $\States = \{\sigma ~|~ \sigma \colon \Vars \To \Vals\}$ denote the \emph{set of program states}. 
We assume an interpretation function $\eval{\:\cdot\:}\colon \DExprs \To \States \To [0,\, 1]$ for probabilistic guards: 
$\eval{\pguard}$ maps each program state to the probability that the guard evaluates to $\true$.
We write $\eval{\neg \pguard}$ as a shorthand for $\lambda \sigma\mydot 1 - \eval{\pguard}(\sigma)$.
E.g. $\eval{x \geq y}(\sigma)$ evaluates with probability 1 to \true if $\sigma(x) \geq \sigma(y)$ and otherwise with probability 1 to \false.
As another example $\eval{\nicefrac 1 2}(\sigma)$ evaluates with probability $\nicefrac 1 2$ to \true and with probability $\nicefrac 1 2$ to \false, regardless of $\sigma$.

We now present the effects of programs in $\pProgs$.
For that, let $\sigma$ be the current program state.
$\SKIP$ has no effect on the program state.
$\ASSIGN{x}{E}$ is an assignment which evaluates expression $E$ in the current program state and assigns this value to variable $x$.
$\COMPOSE{C_1}{C_2}$ is the sequential composition of programs $C_1$ and $C_2$, \ie first $C_1$ is executed, then $C_2$.
$\ITE{\pguard}{C_1}{C_2}$ is a probabilistic conditional branching: $C_1$ is executed with probability $\eval{\pguard}(\sigma)$ and $C_2$ with probability $1-\eval{\pguard}(\sigma)$.
$\WHILEDO{\pguard}{C}$ is a probabilistically guarded while loop: with probability $\eval{\pguard}(\sigma)$ the loop body $C$ is executed followed by a recursive execution of the loop, whereas with probability $1-\eval{\pguard}(\sigma)$ the loop terminates.
\needspace{3\baselineskip}
\begin{example}[Tortoise and Hare~\cite{Chakarov:CAV:13}]\label{ex:program}
The program
\begin{align*}
&\COMPOSE{\COMPOSE{\ASSIGN{t}{30}}{\ASSIGN{h}{0}}}{}\\
&\WHILE~(h \leq t)~\{ \\
&\qquad \COMPOSE{\ASSIGN{t}{t+1}}{} \\
&\qquad \ITE{\nicefrac 1 2}{\ASSIGN{h}{h + 3}}{\SKIP} \}
\end{align*}
illustrates the use of the programming language.
It models a race between a tortoise and a hare ($t$ and $h$ represent their respective positions). 
The tortoise starts with a lead of $30$ and advances one step forward in each round. 
The hare advances three steps or remains still, both with the remaining probability of $\nicefrac 1 2$. 
The race ends when the hare passes the tortoise. 
\hfill$\triangle$
\end{example}

\label{sec:wp}
\section{Non--Negative Weakest Pre--Expectations}
\label{sec:wp-positive}
\noindent
In this section, we recall the standard weakest pre--expectation semantics which acts on non--negative random variables.
When we start a probabilistic program $C$ in some initial state $\sigma$, the final state after termination of $C$ need not be unique due to $C$'s probabilistic nature.
In fact, not even the event of $C$'s termination itself needs to be determined as the program's computation might diverge with a probability that is neither 0 nor 1.
So instead of thinking of $C$ as a mapping from initial to final states, we can rather think of $C$ as a mapping from a distribution $\mu_0$ of initial states to a distribution $\sem{C}(\mu_0)$ of final states.
In order to account for non--termination, we do not require the total probability mass of these distributions to sum up to 1 but any probability between 0 and 1 is valid.
The missing probability mass represents then the probability of non--termination.

Given a random variable $f$ mapping program states to \emph{positive reals}, we can ask:
What is the expected value of $f$ after termination of $C$ when the input to $C$ is distributed according to $\mu_0$?
E.g., what is the expected value of $h$ after termination of 
$\If~\bigl(\nicefrac 1 2 \big)~\{\ASSIGN{h}{h + 3}\}~\Else~\{\SKIP\}$
on an initial distribution in which $h$ is $4$ with probability $\nicefrac{2}{3}$ and $h$ is $7$ with probability $\nicefrac{1}{3}$?

In this case, the answer is $6.5$. In general, an answer to this type of questions can be obtained by means of the weakest pre--expectation calculus \cite{DBLP:journals/jcss/Kozen85,mciver,gretz}:
This calculus can be used to reason about the expected value of a random variable after termination of a probabilistic program $C$.
More precisely, the weakest pre--expectation transformer $\wpsymbol[C]$ transforms a given non--negative random variable $f$ into a random variable $g = \wp{C}{f}$, such that for any initial distribution $\mu_0$ the expected value of $f$ under the \emph{final} distribution $\sem{C}(\mu_0)$ coincides with the expected value of $g$ under the \emph{initial} distribution $\mu_0$.\footnote{A correspondence between the operational point of view outlined in Section \ref{sec:language} and the denotational $\wpsymbol$--semantics for probabilistic programs is provided in \cite{gretz}.}
Put formally, we have
\begin{align}
 \Exp{\mu_0}{\wp{C}{f}} ~=~ \Exp{\sem C(\mu_0)}{f}~,\label{eq:def-wp}
\end{align}
where $\Exp{\mu}{h}$ denotes the expected value of a random variable $h$ under distribution $\mu$.
In particular, if the program $C$ is started in a single determined initial state $\sigma$, then the expected value of $f$ after termination of $C$ on input $\sigma$ is given by $\wp{C}{f}(\sigma)$, since $\wp{C}{f}(\sigma) = \Exp{\delta_\sigma}{\wp{C}{f}} = \Exp{\sem C(\delta_\sigma)}{f}$, where $\delta_\sigma$ is the Dirac distribution that assigns the entire probability mass (i.e.\ 1) to the single point $\sigma$.

Notice that $\wpsymbol$ is {not} a distribution transformer per se.
Nevertheless, given some predicate $A$, we can express the probability that $C$ terminates on initial state $\sigma$ in some state satisfying $A$ in terms of $\wpsymbol$ by $\wp{C}{[A]}(\sigma)$, where $[A]$ is the indicator function of predicate $A$.

In the context of the weakest pre--expectation calculus, random variables are usually referred to as \emph{expectations}:
$f$ is called the \emph{post--expectation} and $g = \wp{C}{f}$ is called the \emph{pre--expectation}.\footnote{As the \emph{post}expectation is evaluated in the \emph{final} states and the \emph{pre}expectation is evaluated in the \emph{initial} states.}
The set of expectations is denoted by
\begin{align*}
	\E ~=~ \left\{f ~\middle|~ f \colon \Sigma \rightarrow \Rposinf \right\}~,
\end{align*}
where $\Rposinf = \{r \in \R ~|~ r \geq 0\} \cup \{\infty\}$.
We need the extended real line here, as we want $\wp{C}{f}$ to always be defined for any $C \in \pProgs$ and any $f \in \E$ and the expected value of $f$ after termination of $C$ can easily become infinity.
Notice that for a probabilistic guard $\pguard$, formally both $\eval{\pguard}$ and $\eval{\neg\pguard}$ are expectations as e.g.\ $\eval{\pguard}\colon \Sigma \rightarrow [0,\, 1]$ and so $\eval{\pguard} \in \E$.
\begin{remark}[Positivity of Expectations]
\label{remark:positive}
Since we have restricted ourselves to non--negative random variables in $\E$, the expected value in \mbox{Equation}~(\ref{eq:def-wp}) is always a well--defined positive real or $+\infty$ for any initial distribution. \hfill$\triangle$
\end{remark}
\noindent
The weakest pre--expectation transformer $\wpsymbol[C]$ can be defined by induction on the structure of the program $C$ according to \autoref{table:wp-rules}.
\begin{table}
\caption{Definitions for the $\wpsymbol$  Transformer Acting on $\E$.  
}
\label{table:wp-rules}
\vspace{-1\baselineskip}
\begin{center}
\renewcommand{\arraystretch}{1.5}
\begin{tabular}{l@{\qquad}l}
\hline
$\boldsymbol{C}$ & $\boldwpsymbol\boldsymbol{[C](f)}$\\
\hline\\[-4ex]
%
$\SKIP$ & $f$\\
%
%
%
$\ASSIGN x E$ & $f \subst{x}{E}$\\
$\COMPOSE{C_1}{C_2}$ & $\wp{C_1}{\vphantom{\big(}\wp{C_2}{f}\vphantom{\big(}}$\\
$\ITE{\pguard}{C_1}{C_2}$ & $\eval{\pguard} \cdot \wp{C_1}{f}  + \eval{\neg \pguard} \cdot \wp{C_2}{f}$\\
$\WHILEDO{\pguard}{C'}$ & $\lfp X\mydot \eval{\neg \pguard} \cdot f  + \eval{\pguard} \cdot \wp{C'}{X}$\\
\hline
\end{tabular}
\end{center}
\end{table}
Let us briefly go over these definitions:
$\wpsymbol[\SKIP]$ behaves as the identity since $\SKIP$ does not modify the program state. 
For $\wp{\ASSIGN{x}{E}}{f}$ we return $f\subst{x}{E}$ which is obtained from $f$ by a sort of ``syntactic replacement" of $x$ by $E$, just as in Hoare logic.
More formally, $f\subst{x}{E} = \lambda \sigma\mydot f(\sigma[x \mapsto \sigma(E)])$.
$\wp{\COMPOSE{C_1}{C_2}}{f}$ obtains a pre--expectation for the program $\COMPOSE{C_1}{C_2}$ by applying $\wpsymbol[C_1]$ to
the intermediate expectation obtained from $\wp{C_2}{f}$.
$\wp{\ITE{\pguard}{C_1}{C_2}}{f}$ weights $\wp{C_1}{f}$ and $\wp{C_2}{f}$ according to the probability of the guard evaluating to \true and \false.
Addition and multiplication of expectations is meant pointwise here, so $f + g = \lambda \sigma \mydot f(\sigma) + g(\sigma)$ and $f \cdot g = \lambda \sigma \mydot f(\sigma) \cdot g(\sigma)$.
Before we turn to the definitions for while--loops, let us illustrate the effects of the $\wpsymbol$ transformer by means of an example:
\begin{example}[Truncated Geometric Distribution] 
  \label{example:trunc}
  Consider the following probabilistic program:
  \begin{align*}
C_\mathit{trunc}\boldsymbol{\colon} \;\; 
 &\If~\bigl(\nicefrac 1 2 \bigr)~\{ \SKIP\} ~\Else~ \{\\
 &\qquad \COMPOSE{\ASSIGN{x}{x + 1}}{}\\
&\qquad \If~\bigl(\nicefrac 1 2 \bigr)~\{ \SKIP \}~\Else~\{ \ASSIGN{x}{x + 1} \} \}
\end{align*}
It can be viewed as modeling a truncated geometric distribution: we repeatedly flip a fair coin until observing the first, say, heads or completing the second unsuccessful trial. 
Suppose we want to know the expected value of $x$.
Then we can calculate this by calculating $\wp{C_\mathit{trunc}}{x}$ as follows:\footnote{We have overloaded the notation $x$ that actually denotes the program variable $x$ to the expectation $\lambda \sigma \mydot \sigma(x)$ for the sake of readability.}
\begin{align*}
&\wp{\stmt_\mathit{trunc}}{x}\\
&=~	\frac{1}{2} \cdot \wp{\SKIP}{x} + \frac{1}{2} \cdot
\wp{\COMPOSE{\ldots}{\ldots}}{x}\\
&=~\frac{x}{2}  + \frac{1}{2} \cdot \wpsymbol[\ASSIGN{x}{x + 1}] \left( \frac{1}{2} \cdot \wp{\SKIP}{x}\right.\\
&\qquad\qquad\qquad\qquad\qquad\qquad~\left.{} + \frac{1}{2} \cdot \wp{\ASSIGN{x}{x+1}}{x} \right) \\
&=~	\frac{x}{2} + \frac{1}{2} \cdot \wpsymbol[\ASSIGN{x}{x + 1}] \left( \frac{x}{2} + \frac{x+1}{2}\right) \\
&=~	\frac{x}{2} + \frac{x+1}{4} + \frac{x + 2}{4} ~=~ x + \frac{3}{4}
\end{align*}
Therefore, the expected value of $x$ after execution of $\stmt_\mathit{trunc}$ is $x + \nicefrac{3}{4}$, where $x + \nicefrac{3}{4}$ is to be evaluated in the initial state in which $\stmt_\mathit{trunc}$ is started.
\hfill$\triangle$
\end{example}
\noindent
We now turn to weakest pre--expectations of while--loops.
While the calculation of $\wpsymbol$ in the above example was straightforward as the program $\stmt_\mathit{trunc}$ is loop--free, $\wpsymbol$ of while--loops is defined using fixed point techniques.
For that, we need a complete partial order $\left(\E,\, \leq\right)$ which is given by
\begin{align*}
	f ~\leq~ g \quad\text{iff}\quad \forall \, \sigma\colon f(\sigma) ~\leq~ g(\sigma)~.
\end{align*}
\normalsize
The bottom element of this complete partial order is given by the constantly zero expectation $\zero = \lambda \sigma\mydot 0$.
The supremum is taken pointwise, so for any subset $D\subseteq \E$, $\sup D = \lambda \sigma\mydot \sup_{f \in D} f(\sigma)$.
Notice that this pointwise supremum always exists as any bounded set of real numbers has a supremum and $+\infty$ is a valid supremum of unbounded sets.
Thus $\left(\E,\, \leq\right)$ is indeed a \emph{complete} partial order with bottom element $\zero$. (It is even a complete lattice.)

Using this complete partial order, the weakest pre--expectation of a while--loop $\WHILEDO{\pguard}{C'}$ is then given in terms of the least fixed point of a special transformer $\charwp{\pguard}{C'}{f}\colon \E \To \E$ constructed from the loop guard $\pguard$, the postexpectation $f$, and the $\wpsymbol$ transformer of the loop body $\wpsymbol[C']$ (see \autoref{table:wp-rules}).
The transformer $\charwp{\pguard}{C'}{f}$ is given by
\begin{align*}
	\charwp{\pguard}{C'}{f}(X) ~=~ \eval{\neg \pguard} \cdot f + \eval{\pguard} \cdot \wp{C'}{X}~.
\end{align*}
We call this transformer $\charwp{\pguard}{C'}{f}$ the \emph{characteristic functional} of $\WHILEDO{\pguard}{C'}$ with respect to post--expectation $f$.
The existence of the least fixed point of $\charwp{\pguard}{C'}{f}$ is ensured by a standard denotational semantics argument (see \eg \cite[Ch.~5]{Winskel:1993}), namely Scott--continuity (or simply continuity) of $\charwp{\pguard}{C'}{f}$ which follows from continuity of $\wpsymbol[C']$.
By completeness of the partial order $\left(\E,\, \leq\right)$ and continuity of the transformer $\charwp{\pguard}{C'}{f}$, the Kleene Fixed Point Theorem \cite{lassez1982fixed,abramsky1994domain} gives an even stronger result than mere existence of a least fixed point.
It states that this least fixed point can be constructed in $\omega$ steps by iterated application of $\charwp{\pguard}{C'}{f}$ to the least element $\zero$, i.e.\ 
\begin{align*}
	\lfp \charwp{\pguard}{C'}{f} ~=~ \sup_{n \in \mathbb N} \charwpn{\pguard}{C'}{f}{n}(\zero)~,
\end{align*}
where $\charwpn{\pguard}{C'}{f}{n}$ stands for $n$-fold application of $\charwp{\pguard}{C'}{f}$ to its argument.
As mentioned, this result holds only for continuous functions.
Continuity of $\wpsymbol[C]$ can be shown by structural induction on the structure of $C$ in case $C$ is not a loop and fixed point induction in case that $C$ is a loop.
Besides continuity, the $\wpsymbol$ transformer enjoys several other useful properties:
\begin{theorem}[Properties of $\wpsymbol$ Acting on $\E$ \cite{DBLP:journals/jcss/Kozen85,mciver,jones}]\label{thm:wp-prop}
  For any program $\stmt \in \appProgs$ the following properties hold:
\begin{enumerate}
	\item[\textnormal{\textbf{(1)}}] 
	\textbf{Continuity:} For any subset of expectations $D \subseteq \E$:
	\begin{align*}
		\wp{\stmt}{\sup D} ~=~ \sup_{f \in D} \wp{\stmt}{f}
	\end{align*}
	\item[\textnormal{\textbf{(2)}}] 
	\textbf{Monotonicity:} For any two expectations $f, g\in \E$:
	\begin{align*}
		f ~\leq~ g \quad\text{implies}\quad \wp{\stmt}{f} ~\leq~ \wp{\stmt}{g}
	\end{align*}
	\item[\textnormal{\textbf{(3)}}] 
	\textbf{Linearity:}  For any two expectations $f, g\in \E$ and any constant $r \in \Rpos$:
	\begin{align*}
		\wp{\stmt}{f + r \cdot g} ~=~ \wp{\stmt}{f} + r \cdot \wp{\stmt}{g}
	\end{align*}
	%
	%
	\item[\textnormal{\textbf{(4)}}] 
	\label{thm:wp-prop:4}
	\textbf{Upper Loop Invariants:} For any expectation $I \in \E$:
	\begin{align*}
		\charwp{\pguard}{C}{f}(I) ~\leq~ I ~~\text{implies}~~ \wp{\WHILEDO{\pguard}{C}}{f} \leq I
	\end{align*} 
	\item[\textnormal{\textbf{(5)}}] 
	\label{thm:wp-prop:5}
	\textbf{Lower Loop Invariants:} For any sequence of expectations $(I_n)_{n \in \mathbb N} \subseteq \E$:
	\begin{align*}
		&I_0 ~\leq~ \charwp{\pguard}{C}{f}(\zero) \quad \text{and} \quad I_{n+1} ~\leq~ \charwp{\pguard}{C}{f}(I_n)\\
		&\qquad\text{implies}\qquad \sup_{n \in \mathbb N} I_n ~\leq~ \wp{\WHILEDO{\pguard}{C}}{f}
	\end{align*} 
\end{enumerate}
\end{theorem}
\noindent
We saw that $\wpsymbol$ is well--defined and enjoys several useful properties if we deal only with positive expectations (recall Remark~\ref{remark:positive}).
When dealing with expected values of mixed--sign random variables, things become much more intricate, even in classical probability theory where no computational aspects are considered. 
In the next section, we show how the $\wpsymbol$ calculus can be extended to act on mixed--sign expectations.
\label{sec:wp-negative}
\section{Integrability--Witnessing Expectations}

\noindent
In this section we outline some problems that occur when dealing with mixed--sign expectations and present our idea on how to circumvent them by incorporating a mechanism that keeps track of the integrability of the expectations.

\subsection{Convergence and Definedness Issues}
\label{sec:issues}

\noindent
So far we had our $\wpsymbol$--transformer act on the set $\E$ of po{\-}si{\-}tive valued expectations.
For expectations that may also take negative values, the expected value after program termination might not be defined for different reasons.
In the following, we present two problematic examples.

\paragraph{Indefinite Divergence}
As a first example, we adopt a counterexample from McIver \& Morgan \cite{mciver}: Consider the mixed--sign random variable $f = (-2)^x$.
The expected value of $f$ after execution of $C_\mathit{geo}$, given by
\begin{align*}
	C_\mathit{geo}\boldsymbol{\colon}\quad&\COMPOSE{\ASSIGN{x}{1}}{}\WHILE (\nicefrac{1}{2})\{ \ASSIGN{x}{x+1} \}~,
\end{align*}
on
an arbitrary initial state is described by the series\footnote{$\sum_{v \in \Vals} \mathsf{Pr}_{\eval{C_\mathit{geo}}(\delta_\sigma)}(v) \cdot f(v) = \sum_{i = 1}^{\infty} \frac{(-2)^i}{2^i}$.}
\begin{align*}
	S~=~ \sum_{i = 1}^{\infty} \frac{(-2)^i}{2^i} ~=~ -1 + 1 - 1 + 1 - 1 + \cdots~,
\end{align*}
which is indefinitely divergent, i.e.\ it neither converges to any real value nor does it tend to $+\infty$ or $-\infty$.
Furthermore, the summands of this series can be reordered in such ways that the series tends to $+\infty$ or that it tends to $-\infty$.
In any case, there exists no meaningful and in particular no \emph{unique} expected value of $f$ and thus no classical pre--expectation $\wp{C_\mathit{geo}}{f}$.

If we were to naively apply the standard weakest pre--expectation calculus, we would first obtain a pre--expectation for the loop, by constructing the characteristic functional
\begin{align*}
	F(X) ~\coloneqq~ \charwp{\nicefrac{1}{2}}{\ASSIGN{x}{x+1}}{(-2)^x}(X) ~=~ \frac{(-2)^x}{2} + \frac{X \subst{x}{x+1}}{2}
\end{align*}
and then doing fixed point iteration, i.e.\ iteratively apply $F$ to $\zero$. In doing so, we get the sequence
\begin{align*}
	F(\zero) ~=~ &\frac{(-2)^x}{2} \\
	F^2(\zero) ~=~ &\frac{(-2)^x}{2} + \frac{(-2)^{x+1}}{4} \\
	F^3(\zero) ~=~ &\frac{(-2)^x}{2} + \frac{(-2)^{x+1}}{4}  + \frac{(-2)^{x+2}}{8}\\
	\intertext{and so on. Notice, that the sequence $(F^n(\zero))_{n \in \mathbb N}$ is \emph{not} monotonically increasing, so iteratively applying $F$ to $\zero$ does not yield an ascending chain. If we nevertheless took the limit of this sequence---naively assuming it exists---, we would get}
	F^{\omega}(\zero) ~=~ &\sum_{i=0}^{\omega} \frac{(-2)^{x + i}}{2^{i+1}} ~.
\end{align*}
Finally, we have to apply the $\wpsymbol$--semantics of the assignment preceding the while--loop to $F^\omega(\zero)$, i.e.\ we have to calculate $\wp{\ASSIGN{x}{1}}{F^{\omega}(\zero)}$, which gives
\begin{align*}
	\sum_{i=0}^{\omega} \frac{(-2)^{i + 1}}{2^{i+1}} ~=~ -1 + 1 - 1 + 1 - 1 + \cdots~.
\end{align*}
This is not well--defined and hence we see that the standard $\wpsymbol$ cannot be applied to this example as is.

\paragraph{Non--Absolute Convergence}
As a second example, consider the expected value of the mixed--sign random variable $f' = \nicefrac{(-2)^x}{x}$ after executing $C_\mathit{geo}$.
It is described by the series\footnote{$\sum_{v \in \Vals} \mathsf{Pr}_{\eval{C_\mathit{geo}}(\delta_\sigma)}(v) \cdot f'(v) = \sum_{i = 1}^{\infty} \frac{(-2)^i}{2^i \cdot i}$.}
\begin{align*}
	S'~=~ \sum_{i = 1}^{\infty} \frac{(-2)^i}{2^i \cdot i} ~=~ -1 + \frac{1}{2} - \frac{1}{3} + \frac{1}{4} - \frac{1}{5} + \cdots~.
\end{align*}
This series in this particular ordering converges to $- \ln(2)$.

Again, if we were to naively apply the standard weakest pre--expectation calculus, we would first obtain a pre--expectation for the loop, by constructing the characteristic functional
\begin{align*}
	F'(X) ~\coloneqq~ \charwp{\nicefrac{1}{2}}{\ASSIGN{x}{x+1}}{\nicefrac{(-2)^x}{x}}(X) ~=~ \frac{(-2)^x}{2 \cdot x} + \frac{X \subst{x}{x+1}}{2}
\end{align*}
and then do fixed point iteration, i.e.\ iteratively apply $F'$ to $\zero$.
This yields
\begin{align*}
	F'(\zero) ~=~ &\frac{(-2)^x}{2 \cdot x} \\
	F'^2(\zero) ~=~ &\frac{(-2)^x}{2 \cdot x} + \frac{(-2)^{x+1}}{4 \cdot (x+1)} \\
	F'^3(\zero) ~=~ &\frac{(-2)^x}{2 \cdot x} + \frac{(-2)^{x+1}}{4 \cdot (x+1)}  + \frac{(-2)^{x+2}}{8 \cdot (x+2)} \\
	\intertext{and so on. Notice that, again, the sequence $(F'^n(\zero))_{n \in \mathbb N}$ is \emph{not} monotonically increasing, so iteratively applying $F'$ to $\zero$ does not yield an ascending chain. If we nevertheless take the limit of this sequence---again just assuming it exists---, we get}
	F'^{\omega}(\zero) ~=~ &\sum_{i=0}^{\omega} \frac{(-2)^{x + i}}{2^{i+1} \cdot (x + i)} ~.
\end{align*}
Finally, we have to calculate $\wp{\ASSIGN{x}{1}}{F'^{\omega}(\zero)}$, which gives
\begin{align*}
	 \sum_{i=0}^{\omega} \frac{(-2)^{i + 1}}{2^{i+1} \cdot (1 + i)} ~=~ -1 + \frac{1}{2} - \frac{1}{3} + \frac{1}{4} - \frac{1}{5} + \cdots~,
\end{align*}
and converges to $-\ln(2)$.

The reason that this example is nevertheless problematic is that by the well--known Riemann Series Theorem \cite{riemann}, the series $S'$ can be reordered in such a fashion that the series converges to any value in $\mathbb R \cup \{-\infty,\, +\infty\}$.
This is because the series does converge but not absolutely.
A series $\sum_{i=0}^\infty a_i$ is said to \emph{converge absolutely} if $\sum_{i=0}^\infty |a_i|$ converges.
If a series is absolutely convergent, then the series is also unconditionally convergent, meaning that the series converges to a unique value regardless of how the summands are ordered.
If, however, a series converges non--absolutely, then the Riemann Series Theorem states that its summands can always be reordered in such a way that the series converges to an arbitrary value or that it tends to $+\infty$ or that it tends to $-\infty$.

This behavior of non--absolutely convergent series under reordering is highly undesirable for expected values since the outcomes of random events are only assigned a probability, and there exists no natural ordering of the summands in which their weighted masses should be summed up to an expected value.
This is the reason why in classical probability theory the expected value $\Exp{\mu}{f}$ of a mixed--sign random variable $f$ is only defined if $\Exp{\mu}{|f|} < \infty$, because that condition is exactly what ensures absolute convergence of the series representing $\Exp{\mu}{f}$.
Next, we investigate how to incorporate the notion of absolute convergence into a new notion for mixed--sign expectations.

\subsection{Integrability--Witnessing Expectations}

\noindent
If a random variable $f$ fulfills the condition $\Exp{\mu}{|f|} < \infty$, then $f$ is called \emph{integrable}.
Our goal is to formally incorporate the bookkeeping whether $f$ is integrable or not into the objects on which a new weakest pre--expectation calculus acts in order to obtain a sound calculus for mixed--sign expectations.
The first thing on our path to this goal is to alter our expectation space to allow for random variables to evaluate to both positive and negative reals.
\begin{definition}[Mixed--Sign Expectations]
	The set $\Epm$ of mixed--sign expectations (or simply expectations) is given by
	\begin{align*}
		\Epm ~=~ \{ f ~|~ f \colon \States \To \R \}~. \tag*{$\triangle$}
	\end{align*}
\end{definition}
\noindent
Notice that we have dropped the $\infty$ element from the co-domain of an expectation, since if $f$ is integrable, then the expected value of $f$ is finite anyway.

Next, we present our integrability bookkeeping approach.
The idea for keeping track of the integrability of an expectation $f$ is to keep a pair of expectations $(f,\, g)$ where $g$ is a \emph{non--negative} expectation that bounds $|f|$.
We call such a pair an \emph{integrability--witnessing pair}.
The idea is that pre--expectations are computed for both components simultaneously.
\begin{definition}[Integrability--Witnessing Pairs]
\label{def:iwe}
	The set $\Pairs$ of integrability--witnessing pairs is defined as a set of pairs
	\begin{align*}
		\Pairs ~=~ \big\{\Pair{f}{g} ~\big|~ f \in \Epm,\, g \in \E,\, |f| \leq g \big\}~.
	\end{align*}
	We define addition of two integrability--witnessing pairs by $\Pair{f}{g} + \Pair{f'}{g'} = \Pair{f + f'}{g + g'}$, a scalar multiplication by $c \cdot \Pair{f}{g} = \Pair{c \cdot f}{|c| \cdot g}$ for $c \in \R$, and a multiplication by $h \cdot \Pair{f}{g} = \Pair{h \cdot f}{|h| \cdot g}$, for $h \in \Epm$.
	\hfill$\triangle$
\end{definition}
\noindent
Next, we would like to define an ordering on integrability--witnessing pairs.
We would like to compare pairs component{\-}wise, i.e.\ $\Pair{f}{g}$ should be less or equal $\Pair{f'}{g'}$ if both $f \leq f'$ and $g \leq g'$.
This would naturally extend the complete partial order $\leq$ on $\E$ to $\Pairs$.
There is, however, a catch:

Recall that the intuition behind a pair $\Pair{f}{g}$ is that whenever the expected value of $g$ is finite, then the expected value of $|f|$, too, is finite by monotonicity of the expected value operator.
If the expected value of $g$ is infinity, however, then the expected value of $f$ cannot be ensured to be defined.
(In particular, if $g = |f|$, then the expected value of $f$ should definitely be undefined.)
Therefore, if $g'$ is the pre--expectation of $g$ and for a state $\sigma \in \States$ we have $g'(\sigma) = \infty$, then we should not care about the pre--expectation of $f$ in state $\sigma$ since definedness cannot be ensured.
This consideration should be reflected in our order on $\Pairs$:
For states where the second component evaluates to $\infty$, the first component should not be compared. 
This gives rise to the following definition:
\begin{definition}[The Quasi--Order $\precsim$ on $\Pairs$]
	A quasi--order $(\Pairs,\, {\precsim})$ is given by 
	\belowdisplayskip=0pt
	\begin{align*}
		\Pair{f}{g} ~\precsim~ (f',\, g')
	\end{align*}
	\normalsize
	iff for all $\sigma \in \States$,
	\begin{align*}
		&g'(\sigma) \neq \infty \quad\text{implies}\quad f(\sigma) \leq f'(\sigma) ~\text{and}~ g(\sigma) \leq g'(\sigma)~. \tag*{$\triangle$}
	\end{align*}
\end{definition}
\noindent
In contrast to a partial order which is reflexive, transitive and antisymmetric, in a quasi--order the requirement of antisymmetry is dropped.
Notice that, indeed, $\precsim$ is only a quasi--order since we can have two integrability--witnessing pairs $\Pair{f}{g}$ and $\Pair{f'}{g'}$ such that for some state $\sigma \in \States$ we have $g(\sigma) = \infty = g'(\sigma)$, but $f(\sigma) \neq f'(\sigma)$, and so $\Pair{f}{g} \neq \Pair{f'}{g'}$.
Still $\Pair{f}{g}$ and $\Pair{f'}{g'}$ compare in both directions, so we have $\Pair{f}{g} \precsim \Pair{f'}{g'}$ and $\Pair{f}{g} \succsim \Pair{f'}{g'}$, but not $\Pair{f}{g} = \Pair{f'}{g'}$.
This leads us to finding that $\precsim$ is \emph{not antisymmetric}.

On the other hand, two integrability--witnessing pairs $\Pair{f}{g}$ and $\Pair{f'}{g'}$, for which $f(\sigma) \neq f'(\sigma)$ holds only for those states in which $g(\sigma) = \infty = g'(\sigma)$, should be considered \emph{equivalent}, even though they are not equal.
This is because for states $\sigma$ in which $g(\sigma) = \infty = g'(\sigma)$, the evaluations of $f(\sigma)$ and $f'(\sigma)$ are ignored since integrability is not ensured.
Consequently, we need a notion of equivalence of integrability--witnessing pairs:
\begin{definition}[Integrability--Witnessing Expectations]
	The quasi--order $\precsim$ induces a canonical \textnormal{\cite{abramsky1994domain}} 
	equivalence relation ${\approx}$, given by $\approx \, = {\precsim} \cap {\succsim}$, i.e.\ 
	\begin{align*}
		\Pair{f}{g} ~\approx~ (f',\, g') \qquad
	\end{align*}
	\normalsize
	iff for all $\sigma \in \States$,
	\begin{align*}
		&g(\sigma) ~\neq~ \infty \quad\text{or}\quad g'(\sigma) ~\neq~ \infty \\
		&\qquad\text{implies}\qquad f(\sigma) ~=~ f'(\sigma) \quad\text{and}\quad g(\sigma) ~=~ g'(\sigma)~.
	\end{align*}
	We denote by $\Lbag \Pair{f}{g} \Rbag$ or simply $\EqPair{f}{g}$ the equivalence class of $\Pair{f}{g}$ under $\approx$ and call such an equivalence class an \emph{integrability--witnessing expectation}.
	We denote by $\EqPairs$ the set of integrability--witnessing expectations, i.e.\ the set of equivalence classes of $\approx$.
	\hfill$\triangle$
\end{definition}
\noindent
Intuitively, an equivalence class $\EqPair{f}{g}$ can be thought of as a particular pair $\Pair{f}{g}$ such that $g$ maps each state either to a non--negative real number or $\infty$ and $f$ maps each state that is not mapped to $\infty$ by $g$ to a real number.

Notice that we call the equivalence classes and not the pairs ``expectations" as we consider $\EqPairs$ and not $\Pairs$ to be a suitable domain to perform computations on and thus we consider  $\EqPairs$ to be the mixed--sign counterpart to $\E$.
Next, we define a \emph{partial order on the equivalence classes}:
\begin{definition}[The Partial Order on $\EqPairs$]
	The quasi--order $\precsim$ on the set  $\Pairs$ of integrability--witnessing pairs induces a canonical \textnormal{\cite{abramsky1994domain}} partial order $\sqsubseteq$ on the set $\EqPairs$ of integrability--witnessing expectations by
	\begin{align*}
		\EqPair{f_1}{g_1} ~\sqsubseteq~ \EqPair{f_2}{g_2} \quad\text{iff}\quad \Pair{f_1}{g_1} ~\precsim~ \Pair{f_2}{g_2}~. \tag*{$\triangle$}
	\end{align*}
\end{definition}
\noindent
As for an intuitive interpretation of this partial order, we note that if $\EqPair{f_1}{g_1} \sqsubseteq \EqPair{f_2}{g_2}$ holds, then we have $f_1'(\sigma) = f_1(\sigma) \leq f_2(\sigma) = f_2'(\sigma)$ for all $\Pair{f_1'}{g_1'} \in \EqPair{f_1}{g_1}$, $\Pair{f_2'}{g_2'} \in \EqPair{f_2}{g_2}$, and all states $\sigma$ in which $g_2(\sigma) \neq \infty$ holds.
Thus if integrability in $\sigma$ is ensured, the first components compare in $\sigma$, which is the comparison we are mainly interested in.

The partial order $\sqsubseteq$ on $\EqPairs$ is complete in the sense that every \emph{non--empty} subset $D \subseteq \EqPairs$ has a supremum $\sup D = \EqPair{\hat f}{\hat g}$ given by
\begin{align*}
	\hat g(\sigma) = &\sup \big\{g(\sigma) ~\big|~ \Pair{f}{g} \in \EqPair{f}{g} \in D\big\} \\[2ex]
	\hat f(\sigma) = &\begin{cases}
		\sup \{f(\sigma) ~|~ \Pair{f}{g} \in \EqPair{f}{g} \in D\}, &\text{if $\hat g(\sigma) \neq \infty$},\\[1ex]
		0, &\text{otherwise,\footnotemark}
	\end{cases}
\end{align*}
where $\infty$ is assumed to be a valid supremum for $\hat g(\sigma)$. \footnotetext{\label{footnotelabel}Notice that this 0 is an arbitrary choice of a value in $\R$ since any $\Pair{\hat f'}{\hat g}$, where $\hat f'(\sigma) \neq 0$ for any $\sigma \in \States$ with $\hat g(\sigma) = \infty$, is in the same equivalence class as $\Pair{\hat f}{\hat g}$.} 

An unfortunate fact about the partial order $(\EqPairs,\, \sqsubseteq)$ is that it has \emph{no least element}.
In particular $\EqPair{\ctert{0}}{\ctert{0}}$ is not a least element of $\EqPairs$ since, for example $\EqPair{\ctert{0}}{\ctert{0}} ~\not\sqsubseteq~ \EqPair{\ctert{-1}}{\ctert{1}}$, where $\ctert{-1} = \lambda \sigma \mydot {-}1$ and $\ctert{1} = \lambda \sigma \mydot 1$.
This fact prevents us from applying the Kleene Fixed Point Theorem---as is typically done in $\wpsymbol$--calculi---in our later development.

In the next section, we investigate a weakest pre--expectation calculus acting on integrability--witnessing expectations.

\section{Mixed--Sign Weakest Pre--Expectations}

%

\noindent
We now develop a weakest pre--expectation calculus acting on integrability--witnessing expectations.
For that we first observe that certain operations on an integrability--witnessing pair $\Pair{f}{g}$ preserve $\approx$--equivalence and thus lifting this operation to the integrability--witnessing expectation $\EqPair{f}{g}$ can be done by performing the operation on the representative $\Pair{f}{g}$ and then taking the equivalence class of the resulting pair.

E.g., the assignment $\ASSIGN{x}{E}$ preserves $\approx$--equivalence, since if $\Pair{f}{g} \approx \Pair{f'}{g'}$ then for all $\sigma \in \States$ we have
\begin{align*}
		&g(\sigma) \neq \infty \quad\text{or}\quad \infty \neq g'(\sigma)\\
		&\qquad\textnormal{implies}\qquad f(\sigma) = f'(\sigma) ~\quad\textnormal{and}\quad~ g(\sigma) = g'(\sigma)~.
\end{align*}
But then this is in particular true for all updated states of the form $\sigma[x \mapsto \sigma(E)]$ and thus $\approx$--equivalence is preserved by the assignment, i.e.\ 
\begin{align*}
	&\Pair{f}{g} ~\approx~ \Pair{f'}{g'}\\
	&\quad\text{implies}\quad \Pair{f\subst{x}{E}}{g\subst{x}{E}} ~\approx~ \Pair{f'\subst{x}{E}}{g'\subst{x}{E}}.
\end{align*}
Moreover this allows for defining a transformer
\begin{align*}
	\wpeqp{\ASSIGN{x}{E}}{f}{g} ~=~ \EqPair{f\subst{x}{E}}{g\subst{x}{E}}~.
\end{align*}
Furthermore, one can show that addition, scalar multiplication, and multiplication also preserve $\approx$--equivalence.
This puts us in a position to formally define a weakest pre--expectation transformer acting on $\EqPairs$:
\begin{definition}[The Transformer $\widetilde{\wpsymbol}$]
	The transformer $\widetilde{\wpsymbol}[C] \colon \EqPairs \To \EqPairs$ is defined by induction on the structure of $C$ according to \autoref{table:wp-pairs}.
	\hfill$\triangle$
\end{definition}
\begin{table}
\caption{\vspace{.25em}Definitions for the $\widetilde{\wpsymbol}$ transformer acting on $\EqPairs$. 
}
\label{table:wp-pairs}
\vspace{-1\baselineskip}
\begin{center}
\renewcommand{\arraystretch}{1.5}
\begin{tabular}{l@{\quad}l}
\hline
$\boldsymbol{C}$ & $\boldsymbol{\widetilde{\boldwpsymbol}[C]\Lbag f,\, g\Rbag}$\\
\hline\\[-4ex]
%
$\SKIP$ & $\EqPair{f}{g}$\\
%
%
%
$\ASSIGN x E$ & $\EqPair{f \subst{x}{E}}{g \subst{x}{E}}$\\
$\COMPOSE{C_1}{C_2}$ & $\widetilde{\wpsymbol}[C_1]\big(\wpeqp{C_2}{f}{g}\big)$\\
$\ITE{\pguard}{C_1}{C_2}$ & $\eval{\pguard} \cdot \wpeqp{C_1}{f}{g}  + \eval{\neg \pguard} \cdot \wpeqp{C_2}{f}{g}$\\
$\WHILEDO{\pguard}{C'}$ & $\displaystyle\lim_{n \To \omega} ~~\charwpn{\pguard}{C'}{\EqPair{f}{g}}{n}\EqPair{\zero}{\zero}$\\[2ex]
\hline
\multicolumn{2}{c}{$\vphantom{\Bigg(}\charwp{\pguard}{C'}{\EqPair{f}{g}}\EqPair{X}{Y} ~=~ \eval{\neg \pguard} \cdot \EqPair{f}{g}  + \eval{\pguard} \cdot \wpeqp{C}{X}{Y}$}\\
\hline
\end{tabular}
\end{center}
\end{table}
\noindent
Let us briefly go over these definitions:
Just like $\wpsymbol[\SKIP]$, $\wpeqpsymbol[\SKIP]$ is an identity since $\SKIP$ does not modify the program state. 
$\wpeqp{\ASSIGN{x}{E}}{f}{g}$ takes a representative $\Pair{f}{g} \in \EqPair{f}{g}$, performs the assignment $\ASSIGN{x}{E}$ on both components to obtain $\Pair{f \subst{x}{E}}{g \subst{x}{E}}$ and then returns the corresponding equivalence class $\EqPair{f \subst{x}{E}}{g \subst{x}{E}}$.
As described earlier, assignments preserve $\approx$--equivalence, so doing the update on the representative is a sound and sufficient course of action.

$\wpeqp{\COMPOSE{C_1}{C_2}}{f}{g}$ obtains a pre--expectation for $\COMPOSE{C_1}{C_2}$ by applying $\wpeqpsymbol[C_1]$ to the intermediate integrability--witnessing expectation obtained from $\wpeqp{C_2}{f}{g}$.
$\wpeqp{\ITE{\pguard}{C_1}{C_2}}{f}{g}$ weights $\wpeqp{C_1}{f}{g}$ and $\wpeqp{C_2}{f}{g}$ according to the probability of the guard $\pguard$ evaluating to \true and \false by multiplication and addition on integrability--witnessing expectations (see \mbox{Definition~\ref{def:iwe}}).

Before we turn our attention to the definitions of $\wpeqpsymbol$ for while--loops, let us illustrate the effects of the $\wpeqpsymbol$ transformer by means of a variation of Example~\ref{example:trunc}:
\begin{example}[Truncated Alternating Geometric Distribution] 
  \label{example:alttrunc}
  Consider the probabilistic program $C_\mathit{alttrunc}$:
  \begin{align*}
C_\mathit{alttrunc}\boldsymbol{\colon} \;\; 
 &\If~\bigl(\nicefrac 1 2 \bigr)~\{ \SKIP \}~\Else~ \{\\
 &\qquad\COMPOSE{\ASSIGN{x}{-x - 1}}{}\\
&\qquad \If~\bigl(\nicefrac 1 2 \bigr)~\{ \SKIP \}~ \Else~\{ \ASSIGN{x}{-x + 1} \}\}
\end{align*}
  It is a variant of $C_\mathit{trunc}$ from Example~\ref{example:trunc} where the program alternates the sign of $x$ and also alternates the sign of the change in $x$.
  Suppose we want to know the expected value of $x$ after termination of $C_\mathit{alttrunc}$.
  The according integrability--witnessing post--expectation for obtaining an answer to this question is $\EqPair{x}{|x|}$.
  Notice that in this example, the need for mixed--sign random variables arises \emph{not} from some artificially constructed mixed--sign post--expectation but \emph{directly from the program code}.
In order to reason about the expected value of $x$ after termination, we calculate $\wpeqp{C_\mathit{alttrunc}}{x}{|x|}$:
\begin{align*}
&\wpeqp{\stmt_\mathit{alttrunc}}{x}{|x|}\\
&=~\frac{1}{2} \cdot \wpeqp{\SKIP}{x}{|x|} + \frac{1}{2} \cdot \wpeqpsymbol[\ldots;~\ldots]\EqPair{x}{|x|}\\
&=~	\frac{1}{2} \cdot \EqPair{x}{|x|}  + \frac{1}{2} \cdot \wpeqpsymbol[\ldots] \left( \frac{1}{2} \cdot \wpeqp{\SKIP}{x}{|x|} \right. \\[0em]
	&\qquad\qquad\qquad\qquad\qquad \left. {}+ \frac{1}{2} \cdot \wpeqp{\ASSIGN{x}{-x+1}}{x}{|x|} \right) \\
&=~ \bigEqPair{\frac{x}{2}}{\frac{|x|}{2}}  + \frac{1}{2} \cdot \wpeqpsymbol[\ASSIGN{x}{-x - 1}] \left( \frac{1}{2} \cdot \EqPair{x}{|x|} \right.\\[0em]
&\qquad \qquad\qquad \qquad\qquad \qquad \left. {}+ \frac{1}{2} \cdot \EqPair{-x + 1}{|{-}x + 1|} \right) \\
&=~	\bigEqPair{\frac{x}{2}}{\frac{|x|}{2}}  + \bigEqPair{\frac{-x-1}{4}}{\frac{|x+1|}{4}}\\
&\qquad{}  + \bigEqPair{\frac{x + 2}{4}}{\frac{|x + 2|}{4}} \\
&=~	\bigEqPair{\frac{x}{2} + \frac{1}{4}}{\frac{2\cdot |x| + |x +1| +|x+2|}{4}}
\end{align*}
The first observation we can make from this result is that the expected value of $x$ is defined after execution of $\stmt_\mathit{alttrunc}$, since in every initial state we have $\nicefrac{2\cdot |x| + |x +1| +|x+2|}{4} < \infty$.
The second observation we can make is that this expected value is for every initial state given by $\nicefrac{x}{2} + \frac{1}{4}$, which is to be evaluated in the initial state in which $\stmt_\mathit{alttrunc}$ is started.
In particular, the above expression gives the correct expected value, \emph{regardless} of whether the program is started with a positive or negative variable valuation for $x$.
\hfill$\triangle$
\end{example}
\noindent
We now turn towards weakest pre--expectations of \emph{while--loops}.
While the calculation of $\wpeqpsymbol$ in the above example was straightforward as the program $\stmt_\mathit{alttrunc}$ is loop--free, $\wpeqpsymbol$ of while--loops is defined using a limit construct.
For that we first need to formally define what a limit of a sequence of integrability--witnessing expectations, i.e.\ a limit of a sequence of equivalence classes, is.
\begin{definition}[Limits of Sequences in $\EqPairs$]
	Let $\left( \EqPair{f_n}{g_n} \right)_{n \in \mathbb N} \subseteq \EqPairs$ be a sequence in $\EqPairs$. 
	Then 
	\begin{align*}
		\EqPair{f}{g} \quad\text{is a limit of}\quad \left( \EqPair{f_n}{g_n} \right)_{n \in \mathbb N}~,
	\end{align*}
	if there exists a sequence $(\Pair{f_{n}'}{g_{n}'})_{n \in \mathbb N}$ of representatives, i.e.\ for all $n \in \mathbb N$, $\Pair{f_n'}{g_n'} \in \EqPair{f_n}{g_n}$, with
	\begin{align*}
		f(\sigma) ~=~ &\begin{cases}
			\displaystyle \lim_{n \To \omega} f_n'(\sigma)~, &\text{if } \displaystyle\lim_{n \To \omega} g_n'(\sigma) \neq \infty~,\\[2ex]
			0~, &\text{otherwise,\footnotemark{} and}
		\end{cases}\\[1em]
		g(\sigma) ~=~ &\lim_{n \To \omega} g_n'(\sigma)~,
	\end{align*}
	where $\infty$ is assumed to be a valid limit for $g'_{n}(\sigma)$.
	\hfill$\triangle$
\end{definition}
\noindent
The\footnotetext{Notice that this 0 is again an arbitrary choice of a value in $\R$ since any $\Pair{f'}{g}$, where $f'(\sigma) \neq 0$ for any $\sigma \in \States$ with $g(\sigma) = \infty$, is in the same equivalence class as $\Pair{f}{g}$. See also Footnote \ref{footnotelabel}.} intuition behind this definition is that a limit of a sequence in $\EqPairs$ is a pointwise limit (in each state $\sigma \in \States$).

If a limit exists, we note the following:
For each pair in any equivalence class, the second component is unique.
Thus the sequence $(g_n')_{n \in \mathbb N}$ is uniquely determined by $(g_n)_{n \in \mathbb N}$.
	
	Now, if $\lim_{n \To \infty} g_n(\sigma) = \infty$, then the limit in that state $\sigma$ does not depend on the sequence $(f_n')_{n \in \mathbb N}$ and is uniquely determined.
	If on the other hand $\lim_{n \To \infty} g_n(\sigma) \neq \infty$, then for almost all $g_i$ we have $g_i(\sigma) \neq \infty$ and thus also almost all $f_i'$ are uniquely determined by $f_i$.	
All in all this leads to the fact that if a limit of $\left( \EqPair{f_n}{g_n} \right)_{n \in \mathbb N}$ exists, then we can reason about the existence by means of the sequence of representatives $(\Pair{f_{n}}{g_{n}})_{n \in \mathbb N}$.

The $\wpeqpsymbol$--semantics of while--loops is defined as the limit of a sequence of integrability--witnessing expectations, but in order to speak of \emph{the} limit, such limits must be unique if they exist.
This is ensured by the following theorem:
\begin{theorem}[Uniqueness of Limits in $\EqPairs$]
	Let $\left( \EqPair{f_n}{g_n} \right)_{n \in \mathbb N} \subseteq \EqPairs$ and let a limit of that sequence exist.
	Then that limit is unique, i.e.\ if $\EqPair{f}{g}$ and $\EqPair{f'}{g'}$ are both a limit of $\EqPair{f_n}{g_n}$, then $\EqPair{f}{g} = \EqPair{f'}{g'}$.
\end{theorem}
\begin{proof}
	Suppose for a contradiction that $\EqPair{f}{g} \neq \EqPair{f'}{g'}$ are both a limit of the sequence $\left( \EqPair{f_n}{g_n} \right)_{n \in \mathbb N} \subseteq \EqPairs$.	
	Recall that we can reason about such a limit entirely by the sequence of representatives $(\Pair{f_{n}}{g_{n}})_{n \in \mathbb N}$.
	Because of $\EqPair{f}{g} \neq \EqPair{f'}{g'}$ we have $\Pair{f}{g} \not\approx \Pair{f'}{g'}$.
	Hence, there must exist a state $\sigma$ such that
	\begin{align*}
		&g(\sigma) \neq \infty \quad\text{or}\quad g'(\sigma) \neq \infty\\
		& \qquad\text{and}\qquad g(\sigma) \neq g'(\sigma) \quad\text{or}\quad f(\sigma) \neq f'(\sigma)~.
	\end{align*}
	But if that were the case, then for that state $\sigma$ either
	\begin{align*}
		\lim_{n \To \omega} g_n(\sigma) ~=~ g(\sigma) ~&\neq~ g'(\sigma) ~=~ \lim_{n \To \omega} g_n(\sigma),\quad\text{or}\\[1ex]
		\lim_{n \To \omega} f_n(\sigma) ~=~ f(\sigma) ~&\neq~ f'(\sigma) ~=~ \lim_{n \To \omega} f_n(\sigma)
	\end{align*}
	should hold, both of which is a contradiction to the fact that limits of real numbers are unique if they exist.
	Therefore, the assumption $\EqPair{f}{g} \neq \EqPair{f'}{g'}$ cannot be true and the limit of $\left( \EqPair{f_n}{g_n} \right)_{n \in \mathbb N}$ must be unique.
\end{proof}
\noindent
Due to the limit's uniqueness, we are now in a position to write
\begin{align*}
	\lim_{n \To \omega} \EqPair{f_n}{g_n} ~=~ \EqPair{f}{g}~,
\end{align*}
if a limit exists and $\EqPair{f}{g}$ is \emph{the} limit of $\lim_{n \To \omega} \EqPair{f_n}{g_n}$.

Using the limit construct, the $\wpeqpsymbol$ of $\WHILEDO{\pguard}{C'}$ is defined as the limit of iteratively applying the characteristic functional of $\WHILEDO{\xi}{C'}$, given by
\begin{align*}
	\charwp{\pguard}{C'}{\EqPair{f}{g}}\EqPair{X}{Y} ~=~ &\eval{\neg \pguard} \cdot \EqPair{f}{g}  + \eval{\pguard} \cdot \wpeqp{C'}{X}{Y}~,
\end{align*}
to $\EqPair{\zero}{\zero}$. 
Formally, we have defined in \autoref{table:wp-pairs}
\begin{align*}
	\wpeqp{\WHILEDO{\xi}{C'}}{f}{g} ~=~ \lim_{n \To \omega} ~~\charwpn{\pguard}{C'}{\EqPair{f}{g}}{n}\EqPair{\zero}{\zero}~,
\end{align*}
where $\charwpn{\pguard}{C'}{\EqPair{f}{g}}{n}$ denotes the $n$-fold application of $\charwp{\pguard}{C'}{\EqPair{f}{g}}$ to its argument.
This is somewhat similar to the $\wpsymbol$--semantics for non--negative expectations, where we basically have $\wp{\WHILEDO{\xi}{C'}}{f} = \lim_{n \To \omega} ~~\charwpn{\pguard}{C'}{f}{n}(\zero)$, since the Kleene Fixed Point Theorem gives
\begin{align*}
	&\wp{\WHILEDO{\xi}{C'}}{f} \\
	&~=~ \lfp \charwp{\pguard}{C'}{f} \\
	& ~=~ \sup_n \charwpn{\pguard}{C'}{f}{n}(\zero) \tag{Kleene Fixed Point Theorem} \\
	& ~=~  \lim_{n \To \omega} ~~\charwpn{\pguard}{C'}{f}{n}(\zero)~, \tag*{$\left(\text{$\charwpn{\pguard}{C'}{f}{n}(\zero)$ increases monot.\ in $n$}\right)$}
\end{align*}
and ensures existence of this limit.
This, however, works only because $\zero$ is the least element in the complete partial order $(\E,\, {\leq})$ and because of the monotonicity of $\charwp{\pguard}{C'}{f}$ (which follows from continuity) we automatically obtain an ascending chain $\zero \leq \charwp{\pguard}{C'}{f}(\zero) \leq \charwpn{\pguard}{C'}{f}{2}(\zero) \leq \charwpn{\pguard}{C'}{f}{3}(\zero) \leq \cdots$, for which a supremum exists by completeness of the underlying partial order.

In contrast to that, $\EqPair{\zero}{\zero}$ is not the least element in the partial order $(\EqPairs,\, {\sqsubseteq})$ and therefore, the sequence
\begin{align*}
	\left( \charwpn{\pguard}{C'}{\EqPair{f}{g}}{n} \right)_{n \in \mathbb N}\EqPair{\zero}{\zero}
\end{align*}
is not necessarily an ascending chain.
It is because of that, that the Kleene Fixed Point Theorem fails in the context of integrability--witnessing expectations.
We have to ensure the existence of the limit defining the semantics of while--loops by other means.
Obviously, it is desired that this limit always exists in order for $\wpeqpsymbol$ to be a well--defined semantics for \emph{all possible programs} together with \emph{all possible post--expectations}, and indeed, we can establish the following result:
\begin{theorem}[Well--Definedness of $\widetilde{\wpsymbol}$ for While--Loops]
\label{thm:while-well-def}
	Let $\xi \in \DExprs$, $C' \in \pProgs$, and $\EqPair{f}{g} \in \EqPairs$.
	Then the limit
	\begin{align*}
		\wpeqp{\WHILEDO{\xi}{C'}}{f}{g} ~=~ \lim_{n \To \omega} ~~\charwpn{\pguard}{C'}{\EqPair{f}{g}}{n}\EqPair{\zero}{\zero}
	\end{align*}
	exists and hence the $\widetilde{\wpsymbol}$--semantics of any while--loop with respect to any post--expectation is well--defined.
\end{theorem}
The core idea for proving this theorem is adopted from a well--known proof proving that every absolutely convergent series is also convergent.
Let us go over this proof:
If a series $S_{a_i} = \sum_{i=0}^\infty a_i$ converges absolutely this means that $S_{|a_i|} = \sum_{i=0}^\infty |a_i|$ converges to some value $a$, which implies that it does so unconditionally and monotonically since all summands are positive.
This, in turn, implies that $S_{2{\cdot}|a_i|} = \sum_{i=0}^\infty 2{\cdot}|a_i|$ converges unconditionally and monotonically to $2{\cdot}a$.
Since $0 \leq a_i + |a_i| \leq 2 \cdot |a_i|$ holds, we obtain
\begin{align*}
	0 ~\leq~ \sum_{i=0}^\infty |a_i| + a_i ~\leq~ \sum_{i=0}^\infty 2 \cdot |a_i| ~=~ 2 \cdot a~.
\end{align*}
By that we can see that the series $S_{|a_i| + a_i} = \sum_{i=0}^\infty |a_i| + a_i$ is bounded.
Furthermore, since $|a_i| + a_i$ must be positive, $S_{|a_i| + a_i}$ is monotonically increasing and therefore $S_{|a_i| + a_i} = \sum_{i=0}^\infty |a_i| + a_i$ converges unconditionally.
Now, since $S_{a_i}$ is given as the difference of two unconditionally convergent series, namely $S_{a_i} = S_{|a_i| + a_i} - S_{|a_i|}$ the series $S_{a_i}$ must also converge.
This basic idea of ``express $\sum a_i$ as $\sum |a_i| + a_i - \sum |a_i|$" in case that these latter two sums exist, is the underlying principle of the following proof.
\begin{proof}[Proof of Theorem~\ref{thm:while-well-def}]
	The idea of this proof is to show by induction on the nesting depth of while--loops and by induction on $n$ that
	\begin{align*}
		\charwpn{\pguard}{C'}{\EqPair{f}{g}}{n}\EqPair{\zero}{\zero} ~=~ \bigEqPair{\charwpn{\pguard}{C'}{|f| + f}{n}(\zero) - \charwpn{\pguard}{C'}{|f|}{n}(\zero)}{\charwpn{\pguard}{C'}{g}{n}(\zero)}
	\end{align*}
	holds for all $n$ and any $C'$.
	It is then left to show that the limit of the above exists for $n \rightarrow \omega$.
	We can see that the second component of that sequence increases monotonically towards
	\begin{align*}
		\sup_{n \in \mathbb N} \charwpn{\pguard}{C'}{g}{n}(\zero) ~=~ \wp{\WHILEDO{\pguard}{C'}}{g}~.
	\end{align*}
	Then for any state $\sigma$ for which $\wp{\WHILEDO{\pguard}{C'}}{g}(\sigma) < \infty$ holds, we have
	\begin{align*}
		&\sup_{n \in \mathbb N} \charwpn{\pguard}{C'}{|f| + f}{n}(\zero)(\sigma)\\
		&~=~ \wp{\WHILEDO{\pguard}{C'}}{|f| + f}(\sigma) \\
		& ~\leq~ \wp{\WHILEDO{\pguard}{C'}}{2 \cdot |f|}(\sigma) \tag{$\wpsymbol$ monotonic}\\
		& ~\leq~ \wp{\WHILEDO{\pguard}{C'}}{2 \cdot g}(\sigma) \tag{$\wpsymbol$ monotonic} \\
		& ~\leq~ 2 \cdot \wp{\WHILEDO{\pguard}{C'}}{g}(\sigma) \tag{$\wpsymbol$ linear} \\
		& ~<~ 2 \cdot \infty ~=~ \infty~,\quad\text{and}\\[1em]
		&\sup_{n \in \mathbb N} \charwpn{\pguard}{C'}{|f|}{n}(\zero)(\sigma)\\
		& ~=~ \wp{\WHILEDO{\pguard}{C'}}{|f|}(\sigma)\\
		& ~\leq~ \wp{\WHILEDO{\pguard}{C'}}{g}(\sigma) \tag{$\wpsymbol$ monotonic} \\
		& ~<~ \infty~.
	\end{align*}
	Hence, the limit for both $\charwpn{\pguard}{C'}{|f| + f}{n}(\zero)(\sigma)$ and $\charwpn{\pguard}{C'}{|f|}{n}(\zero)(\sigma)$ exists, thus also the limit for $\charwpn{\pguard}{C'}{|f| + f}{n}(\zero)(\sigma) - \charwpn{\pguard}{C'}{|f|}{n}(\zero)(\sigma)$ exists, and therefore $\lim_{n \To \omega} \charwpn{\pguard}{C'}{\EqPair{f}{g}}{n}\EqPair{\zero}{\zero}$
	exists, too.
\end{proof}
\noindent
Let us revisit the two examples we presented in Section~\ref{sec:issues}, i.e.\ the program
\begin{align*}
	C_\mathit{geo}\boldsymbol{\colon}\quad&\COMPOSE{\ASSIGN{x}{1}}{}\WHILE (\nicefrac{1}{2})\{\ASSIGN{x}{x+1}\}~,
\end{align*}
together with post--expectations $f = 2^x$ and $f' = (-2)^x$, respectively.
In the $\wpeqpsymbol$ calculus, the respective pre--expectations are well--defined, namely
\begin{align*}
	&\wpeqp{C_\mathit{geo}}{2^{x}}{|2^x|} ~=~ \EqPair{\zero}{\infty}\\
	&\qquad\text{and} \qquad \wpeqp{C_\mathit{geo}}{(-2)^{x}}{|(-2)^x|} ~=~ \EqPair{\zero}{\infty}~.
\end{align*}
So the pre--expectations of these two examples are perfectly well--defined and therefore these examples are not at all pathological in our presented calculus.

\section{Properties of the $\wpeqpsymbol$--Transformer}

\subsection{Monotonicity}

\noindent
Perhaps the single most important property of the $\wpeqpsymbol$--transformer is monotonicity, as that is what enables compositional reasoning.
Monotonicity is as vital to our calculus as the consequence rule is to Hoare logic.
For instance, it enables to continue reasoning soundly using over--approximations obtained by invariant rules.
Our transformer enjoys this property:
\begin{theorem}[Monotonicity of $\widetilde{\wpsymbol}$]
\label{thm:wp-mixed-mon}
	$\widetilde{\wpsymbol}$ is \emph{monotonic} with respect to $\sqsubseteq$, i.e.\ for all $C \in \pProgs$ and $\EqPair{f}{g}, \EqPair{f'}{g'} \in \EqPairs$,
	\begin{align*}
		&\EqPair{f}{g} ~\sqsubseteq~ \EqPair{f'}{g'}\\
		&\qquad\text{implies}\qquad \wpeqp{C}{f}{g} ~\sqsubseteq~ \wpeqp{C}{f'}{g'}~.
	\end{align*}
\end{theorem}
\begin{proof}
	The proof goes by induction on the structure of $C$.
	Let $\EqPair{f}{g} \sqsubseteq \EqPair{f'}{g'}$.
	All cases are straightforward, except for the while--loop.
	For that, reconsider
	\begin{align*}
		\charwpn{\pguard}{C'}{\EqPair{f}{g}}{n}\EqPair{\zero}{\zero} ~=~ \bigEqPair{\charwpn{\pguard}{C'}{|f| + f}{n}(\zero) - \charwpn{\pguard}{C'}{|f|}{n}(\zero)}{\charwpn{\pguard}{C'}{g}{n}(\zero)}
	\end{align*}
	from the proof of Theorem~\ref{thm:while-well-def}.
	Given that fact, the proof boils down to showing by induction on $n$ that the inequality
	\begin{align*}
	\label{eq:ineq-mon}
		&\charwpn{\pguard}{C'}{|f| + f}{n}(\zero)(\sigma) - \charwpn{\pguard}{C'}{|f|}{n}(\zero)(\sigma) \tag{$\dagger$}\\
		& ~\leq~ \charwpn{\pguard}{C'}{|f'| + f'}{n}(\zero)(\sigma) - \charwpn{\pguard}{C'}{|f'|}{n}(\zero)(\sigma)~,
	\end{align*}
	holds if $\wp{\WHILEDO{\xi}{C'}}{g'}(\sigma) < \infty$ (and therefore by monotonicity also $\wp{\WHILEDO{\xi}{C'}}{g}(\sigma) < \infty$) holds.
	In that case, both 
	\begin{align*}
		\charwpn{\pguard}{C'}{|f|}{n+1}(\zero)(\sigma) \leq \charwpn{\pguard}{C'}{g}{n+1}(\zero)(\sigma) &\leq \left(\lfp\charwp{\pguard}{C'}{g}\right)(\sigma) < \infty
		\intertext{and}
		\charwpn{\pguard}{C'}{|f'|}{n+1}(\zero)(\sigma) \leq \charwpn{\pguard}{C'}{g'}{n+1}(\zero)(\sigma) &\leq \left(\lfp\charwp{\pguard}{C'}{g'}\right)(\sigma) < \infty
	\end{align*} 
	holds, and we can thus rewrite inequality~($\dagger$) as
	\begin{align*}
		&\charwpn{\pguard}{C'}{|f| + f}{n}(\zero)(\sigma) +\charwpn{\pguard}{C'}{|f'|}{n}(\zero)(\sigma)\\
		&  ~\leq~ \charwpn{\pguard}{C'}{|f'| + f'}{n}(\zero)(\sigma) + \charwpn{\pguard}{C'}{|f|}{n}(\zero)(\sigma)~,
	\end{align*}
	and prove that statement by induction on $n$ instead.
	For the induction step, consider the following:
	\begin{align*}
		& \charwpn{\pguard}{C'}{|f| + f}{n+1}(\zero)(\sigma) + \charwpn{\pguard}{C'}{|f'|}{n+1}(\zero)(\sigma) \\
		& ~\leq~ \charwpn{\pguard}{C'}{|f'| + f'}{n+1}(\zero)(\sigma) + \charwpn{\pguard}{C'}{|f|}{n+1}(\zero)(\sigma)\\
	\intertext{iff}
		&\big(\probof{\xi}{\false} \cdot (|f| + f + |f'|)\big)(\sigma) \\
		&{} + \left(\probof{\xi}{\true} \cdot \wp{C'}{\charwpn{\pguard}{C'}{|f| + f}{n}(\zero) + \charwpn{\pguard}{C'}{|f'|}{n}(\zero)}\right)(\sigma) \\
		\leq{}& \big(\probof{\xi}{\false} \cdot (|f'| + f' + |f|)\big)(\sigma) \\
		&{} + \left(\probof{\xi}{\true} \cdot \wp{C'}{\charwpn{\pguard}{C'}{|f'| + f'}{n}(\zero) + \charwpn{\pguard}{C'}{|f|}{n}(\zero)}\right)(\sigma)~, \tag{by definition of characteristic functional and linearity of $\wpsymbol$}
	\end{align*}
	which follows from the induction hypothesis on $n$ and by monotonicity of $\wpsymbol$.
\end{proof}

\subsection{Reasoning about Loops}

\noindent
Whereas reasoning about non--loopy programs is mostly straightforward, reasoning about the $\wpeqpsymbol$ of a loop is more complicated as it involves reasoning about limits of integrability--witnessing expectation sequences.
To help overcoming this difficulty, we present now an invariant--based approach that allows for over--approximating those limits.

We have already seen that the fact
\begin{align*}
		\charwpn{\pguard}{C'}{\EqPair{f}{g}}{n}\EqPair{\zero}{\zero} ~=~ \bigEqPair{\charwpn{\pguard}{C'}{|f| + f}{n}(\zero) - \charwpn{\pguard}{C'}{|f|}{n}(\zero)}{\charwpn{\pguard}{C'}{g}{n}(\zero)}
\end{align*}
from the proof of Theorem~\ref{thm:while-well-def} was vital to showing monotonicity of the $\wpeqpsymbol$--transformer.
It will also allow us to reason about integrability--witnessing pre--expectations through reasoning about standard weakest pre--expectations, which is simpler since we have an easy--to--apply invariant rule for these.

If we take a closer look at the sequence
\begin{align*}
	\left(\bigEqPair{\charwpn{\pguard}{C'}{|f| + f}{n}(\zero) - \charwpn{\pguard}{C'}{|f|}{n}(\zero)}{\charwpn{\pguard}{C'}{g}{n}(\zero)}\right)_{n\in\mathbb{N}}
\end{align*}
we can see that in order to over--approximate the limit of that sequence, we can---simply put---
\begin{enumerate}
	\item
	over--approximate the limit---i.e.\ the supremum---of $\charwpn{\pguard}{C'}{g}{n}(\zero)$,
	\item
	over--approximate the limit---i.e.\ again the supremum---of $\charwpn{\pguard}{C'}{|f| + f}{n}(\zero)$, and
	\item
	under--approximate the limit---once again: the supremum---of $\charwpn{\pguard}{C'}{|f|}{n}(\zero)$.
\end{enumerate}
Notice that these over- and under--approximations are over- and under--ap{\-}prox{\-}i{\-}ma{\-}tions of standard weakest pre--expectations.
Furthermore, recall that by \mbox{Theorem}~\ref{thm:wp-prop}~(4)~and~(5) we have invariant rules for those over-- and under--approximations.
This immediately leads us to the following proof rule for loops:
\begin{theorem}[Loop Invariants for $\boldsymbol{\widetilde{\wpsymbol}}$]
\label{thm:invariants}
	Let $\EqPair{f}{g} \in \EqPairs$, $C' \in \pProgs$, $I, G \in \E$ with $G(\sigma) < \infty$, for all $\sigma \in \States$, and $(H_n)_{n \in \mathbb N} \subseteq \E$. Then $\charwp{\xi}{C'}{g}(G) \leq G$, $\charwp{\xi}{C'}{|f| + f}(I) \leq I$, $H_0 \leq \charwp{\pguard}{C'}{|f|}(\zero)$, and $H_{n+1} \leq \charwp{\pguard}{C'}{|f|}(H_n)$ implies
	\begin{align*}
		\wpeqp{\WHILEDO{\xi}{C'}}{f}{g} ~\sqsubseteq~ \bigEqPair{I - \sup_{n\in \mathbb N} H_n}{2 \cdot G}~.
	\end{align*}
\end{theorem}
\noindent
By similar considerations, we can find a dual theorem for lower bounds, see Appendix~\ref{app:lower-bounds}.
Notice that we have to use $2\cdot G$ in the second component of the over--ap{\-}prox{\-}i{\-}ma{\-}tion of $\wpeqp{\WHILEDO{\xi}{C'}}{f}{g}$.
This is just to ensure that the second component really bounds the absolute value of the first component.
Using $G$ instead might not yield a proper member of $\EqPairs$.
Notice that using $2 \cdot G$ does not effect the integrability--witnessing property of the second component.

Let us now illustrate the use of the loop invariant rule from Theorem~\ref{thm:invariants} by means of a worked example:
\begin{example}[Towards Amortized Expected Run--Time Analysis]
\label{ex:amortized}
Suppose we need to perform an amortized analysis of a randomized data structure by means of a potential function $\Phi$.
Suppose further that a certain operation $\mathit{Op}$ first increases the potential by 1 and thereafter keeps flipping a coin until the first heads.
With every flip of tails though, the potential is decreased by 3.
We can model this situation by means of the following probabilistic program:
\begin{align*}
	C_\mathit{Op}\boldsymbol{\colon}~~~&\COMPOSE{\ASSIGN{\Phi}{\Phi + 1}}{} \WHILE~(\nicefrac{1}{2})~\{ \ASSIGN{\Phi}{\Phi - 3} \}
\end{align*}
Here $\Phi$ represents the change in the potential function $\Phi$.
Notice that the change in potential might very well be positive (in fact with probability $\nicefrac 1 2$) as well as negative, so both possibilities have to be accounted for.

Even though an application of the operation $\mathit{Op}$ might increase the potential, we now want to prove that an application of $\mathit{Op}$ \emph{decreases} the potential \emph{in expectation}.
This amounts to proving that the pre--expectation of $\Phi$ evaluated in any initial state $\sigma$ with $\sigma(\Phi) = 0$ is negative.
For that, we need to calculate 
\begin{align*}
	&\wpeqp{C_\mathit{Op}}{\Phi}{|\Phi|}\\
	&~=~ \wpeqpsymbol[\ASSIGN{\Phi}{\Phi + 1};~\WHILEDO{\nicefrac{1}{2}}{\ldots}]\EqPair{\Phi}{|\Phi|}\\
	&~=~ \wpeqpsymbol[\ASSIGN{\Phi}{\Phi + 1}]\big(\wpeqp{\WHILEDO{\nicefrac{1}{2}}{\ldots}}{\Phi}{|\Phi|}\big)
\end{align*}
So the first thing we need to do is to reason about the pre--expectation of the while--loop.
Appealing to Theorem~\ref{thm:invariants}, we propose following loop invariants
\begin{align*}
	&G = \sum_{i = 0}^{\omega} \frac{|\Phi - 3 \cdot i|}{2^{i+1}}, ~\quad~ I = \sum_{i = 0}^{\omega} \frac{|\Phi - 3 \cdot i|}{2^{i+1}} + \Phi - 3,\\
	& ~\qquad\text{and}\qquad~ H_n = \sum_{i = 0}^{n} \frac{|\Phi - 3 \cdot i|}{2^{i+1}}~.
\end{align*}
Indeed, one can verify that these loop invariants satisfy the preconditions of Theorem~\ref{thm:invariants}.
Furthermore, we observe that $\sup_{n \in \mathbb N} H_n = G$ holds.
Applying \mbox{Theorem}~\ref{thm:invariants} therefore yields
\begin{align*}
	\wpeqp{\WHILEDO{\nicefrac{1}{2}}{\ASSIGN{\Phi}{\Phi - 3}}}{\Phi}{|\Phi|} \sqsubseteq \EqPair{I - G}{2 \cdot G}~.
\end{align*}
Because $G$ and $I$ are absolutely convergent for any valuation of $\Phi$ (e.g.\ by the ratio test), we can calculate $I - G = \Phi - 3$, and so we get
\begin{align*}
	\wpeqp{\WHILEDO{\nicefrac{1}{2}}{\ASSIGN{\Phi}{\Phi - 3}}}{\Phi}{|\Phi|} \sqsubseteq \EqPair{\Phi - 3}{2 \cdot G}~.
\end{align*}
Since $\wpeqpsymbol$ is monotonic (see Theorem~\ref{thm:wp-mixed-mon}), we can now safely continue our reasoning with the over--approximation $\EqPair{\Phi - 3}{2 \cdot G}$ and calculate
\begin{align*}
	&\wpeqp{\ASSIGN{\Phi}{\Phi {+} 1}}{\Phi {-} 3}{2 G} = \hugeEqPair{\Phi {-} 2}{\sum_{i {=} 0}^{\omega} \frac{|\Phi {+}1 {-} 3 i|}{2^{i{+}1}}}~.
\end{align*}
By that, we get in total an over--ap{\-}prox{\-}i{\-}ma{\-}tion of the sought--after pre--ex{\-}pec{\-}ta{\-}tion $\wpeqp{C_\mathit{Op}}{\Phi}{|\Phi|}$.
If we instantiate the second component of that over--ap{\-}prox{\-}i{\-}ma{\-}tion in an initial state $\sigma$ with $\sigma(\Phi) = 0$, we get
\begin{align*}
	\sum_{i = 0}^{\omega} \frac{|0 +1 - 3 \cdot i|}{2^{i+1}} ~=~ 3 ~<~ \infty~.
\end{align*}
So the expected value at $\sigma$ was integrable and thus it makes sense to evaluate the first component in $\sigma$ (which is what we are really interested in).
This gives $0 - 2 = -2$ and thus executing $\mathit{Op}$ decreases the potential \emph{in expectation} by 2.
\hfill$\triangle$
\end{example}
\noindent
Notice that the analysis performed as in \mbox{Example}~\ref{ex:amortized} would not be possible using either the deduction rules of PPDL~\cite{DBLP:journals/jcss/Kozen85} or the invariant--based approach of McIver \& Morgan's $\wpsymbol$--calculus~\cite{mciver} off--the--shelf. Instead, a tailor--made argument would be needed for reasoning about the mixed--sign $\Phi$.
For more details on this matter and a more involved worked example, see Appendix~\ref{app:comparison}.

\subsection{Soundness}

\noindent
The last but certainly not least important property that we establish for our $\wpeqpsymbol$ transformer is that it is sound, meaning that if we can establish $\wpeqp{C}{f}{g} = \EqPair{f'}{g'}$ then $f'(\sigma)$ is in fact the expected value of $f$ after termination of $C$ on initial state $\sigma$.
For that, we first generalize the fact established in the proof of Theorem~\ref{thm:while-well-def}:
\begin{lemma}
\label{lem:important}
	Let $\EqPair{f}{g} \in \EqPairs$ with $\wp{C}{g}(\sigma) < \infty$ and $\EqPair{f'}{g'} = \wpeqp{C}{f}{g}$.
	Then $f'(\sigma) = \wp{C}{|f| + f}(\sigma) - \wp{C}{|f|}(\sigma)$.
\end{lemma}
\begin{proof}
	By induction on $C$ using the fact from the proof of Theorem~\ref{thm:while-well-def}. 
\end{proof}
\noindent
In standard probability theory any mixed--sign random variable $f$ can be decomposed into a positive part $\positive{f}  = \lambda \sigma \mydot \max\{0,\, f(\sigma)\} \in \E$ and a negative part $\negative{f} = \lambda \sigma \mydot {-}\min\{0,\, f(\sigma)\} \in \E$, with $f = \positive{f} - \negative{f}$.
Notice that $\positive{f}$ and $\negative{f}$ are both non--negative expectations.
The expected value of $f$ is then defined as the expected value of $\positive{f}$ minus the expected value of $\negative{f}$, i.e.\ as
\begin{align*}
	\Exp{\mu}{f} ~=~ \Exp{\mu}{\positive{f}} - \Exp{\mu}{\negative{f}}~,
\end{align*}
if both $\Exp{\mu}{\positive{f}} < \infty$ and $\Exp{\mu}{\negative{f}} < \infty$.
Using this---sometimes called---\emph{Jordan decomposition} of $f$, we can now establish the following lemma:
\begin{lemma}
\label{lem:presoundness}
	Let $f \in \Epm$, $g \in \E$ with $|f| \leq g$, $\wp{C}{g}(\sigma) < \infty$, $\wpeqp{C}{f}{g}\allowbreak = \EqPair{f'}{g'}$, and $f = \positive{f} - \negative{f}$.
	Then $f'(\sigma) + \wp{C}{\negative{f}}(\sigma) = \wp{C}{\positive{f}}(\sigma)$.
\end{lemma}
\begin{proof}
Consider the following:
\belowdisplayskip=-1\baselineskip
\begin{align*}
	&f'(\sigma) + \wp{C}{\negative{f}}(\sigma) ~=~ \wp{C}{\positive{f}}(\sigma)\\
	\Longleftrightarrow~&\wp{C}{|f| + f}(\sigma) - \wp{C}{|f|}(\sigma) + \wp{C}{\negative{f}}(\sigma)\\
	&~=~ \wp{C}{\positive{f}}(\sigma) \tag{Lemma \ref{lem:important}}\\
	\Longleftrightarrow~&\wp{C}{|f| + f}(\sigma) + \wp{C}{\negative{f}}(\sigma)\\
	&~=~ \wp{C}{\positive{f}}(\sigma) + \wp{C}{|f|}(\sigma) \tag{by $\wp{C}{|f|}(\sigma) \leq \wp{C}{g}(\sigma) < \infty$}\\
	\Longleftrightarrow~&\wp{C}{|f| + f + \negative{f}}(\sigma) ~=~ \wp{C}{\positive{f}+ |f|}(\sigma) \tag{by linearity of $\wpsymbol$}\\
	\Longleftrightarrow~&\wp{C}{|f| + \positive{f}}(\sigma) ~=~ \wp{C}{\positive{f}+ |f|}(\sigma) \tag{by $f = \positive{f} - \negative{f}$ iff $\positive{f} = f + \negative{f}$}\\
	\Longleftrightarrow~&\true
\end{align*}
\normalsize
\end{proof}
\noindent
The soundness of the $\wpeqpsymbol$ transformer follows now almost immediately from Lemma~\ref{lem:presoundness}.
\begin{theorem}[Soundness of $\boldsymbol{\widetilde{\boldwpsymbol}}$]
	Let $f\in \Epm$, $\wp{C}{|f|}(\sigma) < \infty$, and let $\wpeqp{C}{f}{|f|} \allowbreak= \EqPair{f'}{g'}$.
	Then $f'(\sigma)$ is the expected value of $f$ after termination of program $C$ on state $\sigma$, i.e.\ if $f = \positive{f} - \negative{f}$, then
	\begin{align*}
		f'(\sigma) ~=~ \wp{C}{\positive{f}}(\sigma) - \wp{C}{\negative{f}}(\sigma)~.
	\end{align*}
\end{theorem}
\begin{proof}
	In principle by setting $g = |f|$ in Lemma~\ref{lem:presoundness}.
\end{proof}
%
%


\section{Conclusion}
\label{sec:conclusion}
\noindent
We have presented a sound weakest pre--expectation calculus for reasoning about \emph{mixed--sign unbounded} expectations.
With this calculus, a pre--expectation can always be obtained, even in those cases where classical pre--expectations (i.e.\ expected values) do not exist.
We have shown that the semantics of while--loops is always well--defined in terms of a limit of iteratively applying a functional to a zero--element despite the fact that a standard least fixed point argument is not applicable in this context.
For reasoning about loops, we have presented an invariant--based technique and we have shown its applicability to an example inspired by amortized analysis of randomized algorithms.

\section*{Acknowledgements}
\noindent
We would like to thank the anonymous referees on their constructive feedback on an earlier version of this paper.



\bibliographystyle{IEEEtran}
\bibliography{IEEEabrv,literature}
%

\clearpage
\appendix
\section{Appendix}
\subsection{Proof of Expected Run--Time for Kozen's Example}
\label{app:runtimewpproof}

\noindent
The example studied by Kozen~\cite{DBLP:journals/jcss/Kozen85} is given by the following program:
\begin{align*}
	&\ASSIGN{x}{n};~ \ASSIGN{c}{0};\\
	&\WHILE~(x \neq 0)~\{\\
	&\qquad \ITE{\nicefrac{1}{2}}{\SKIP}{\ASSIGN{x}{x - 1}};\\
	&\qquad \ASSIGN{c}{c + 1}\}
\end{align*}
As~\cite{DBLP:journals/jcss/Kozen85}, we assume that all program variables range over the integers and that $n > 0$.
The (asymptotic) expected run--time of this program is given by the expected value of $c$ after termination of the program, as $c$ counts the number of loop iterations.
A detailed proof that $2 n$ is an upper bound of that expected run--time carried out in the $\wpsymbol$--calculus is given in the following:
\begin{proof}
	We start our analysis with post--expectation $c$.
	Appealing to \mbox{Theorem}~\ref{thm:wp-prop}~(4), we look for an upper invariant $I$ for the while--loop, such that
	\begin{align}
		\label{eq:kozeninv}
		\ExpToFun{x = 0} \cdot c + \ExpToFun{x \neq 0} \cdot \wp{\mathit{body}}{I} ~\leq~ I~,
	\end{align}
	e.g.\ $I = \ExpToFun{x \geq 0} (c + 2 x)$.
	Next, we check that $I$ indeed satisfies Equation~(\ref{eq:kozeninv}):
	\begin{align*}
		&\ExpToFun{x = 0} \cdot c + \ExpToFun{x \neq 0} \cdot \wp{\mathit{body}}{I}\\
		&{}= \vphantom{\frac{1}{2}}\ExpToFun{x = 0} \cdot c + \ExpToFun{x \neq 0} \cdot \wp{\mathtt{if}~(\nicefrac{1}{2})~\{\ldots\};~\ASSIGN{c}{c + 1}}{I}\\
		&{}= \vphantom{\frac{1}{2}}\ExpToFun{x = 0} \cdot c  + \ExpToFun{x \neq 0} \cdot \wp{\mathtt{if}\ldots}{\ExpToFun{x \geq 0} (c + 1 + 2 x)}\\
		&{}= \ExpToFun{x = 0} \cdot c + \ExpToFun{x \neq 0} \cdot \frac{1}{2} \big( \ExpToFun{x \geq 0} (c + 1 + 2 x)\\
		&~~\qquad\qquad\qquad\qquad\qquad + \ExpToFun{x - 1\geq 0} (c + 1 + 2 (x - 1)) \big)\\
		&{}= \ExpToFun{x = 0} \cdot c\\
		&\quad~~{} + \frac{1}{2} \big( \ExpToFun{x \geq 1} (c + 1 + 2 x) + \ExpToFun{x \geq 1} (c + 1 + 2 (x - 1)) \big)\\
		&{}= \ExpToFun{x = 0} \cdot c + \frac{1}{2} \ExpToFun{x \geq 1} (c + 1 + 2 x + c + 1 + 2 (x - 1))\\
		&{}= \vphantom{\frac{1}{2}}\ExpToFun{x = 0} \cdot (c + 2x) + \ExpToFun{x \geq 1} (c + 2x)\\
		&{}= \vphantom{\frac{1}{2}}\ExpToFun{x \geq 0} (c + 2 x) ~=~ I ~\leq~ I
	\end{align*}
	By Theorem \ref{thm:wp-prop} (4) we have now established $$\wp{\WHILEDO{x \neq 0}{\mathit{body}}}{c} ~\leq~ \ExpToFun{x \geq 0} (c + 2 x)~,$$ and by monotonicity of $\wpsymbol$, we can proceed our analysis with
	\begin{align*}
		&\wp{\ASSIGN{x}{n};\ASSIGN{c}{0}}{I}\\
		& ~=~ \wp{\ASSIGN{x}{n}}{\ExpToFun{x \geq 0} (0 + 2 x)}  ~=~  \ExpToFun{n \geq 0} \cdot 2 n~.
	\end{align*}
	By assumption $n > 0$, we have proven the run--time bound $2n$.
\end{proof}
\noindent
All calculations in the invariant verification involved only basic arithmetic.
The $\wpsymbol$--style proof that $2n$ is also a lower bound is slightly more involved.
The reader is invited to compare the above reasoning to the reasoning using \mbox{PPDL~\cite[pages 176--177]{DBLP:journals/jcss/Kozen85}}.

\subsection{A More Involved Example}
\label{app:comparison}

\noindent
Consider the following program
\begin{align*}
	&\WHILEDO{\nicefrac{1}{2}}{\ASSIGN{x}{-x-\textsf{sign}(x)}}
\end{align*}
We are interested in the expected value of program variable $x$ and will establish upper bounds for this expected value by means of two methods: first using the approach present in this paper and second by a Jordan--decomposition--based approach.
We will see that the latter is more involved.

\subsubsection{Analysis Using Integrability--Witnessing Expectations}
\label{app:comparison:iwe}
We perform an analysis of
\begin{align*}
\wpeqp{\WHILEDO{\nicefrac{1}{2}}{C'}}{x}{|x|}~,
\end{align*}
where we denote by $C'$ the program $\ASSIGN{x}{-x-\textsf{sign}(x)}$ for the sake of readability.
Appealing to Theorem~\ref{thm:invariants}, we propose the following loop invariants:
\begin{align*}
	&G = |x| + 1, ~~ I = |x| + \ExpToFun{x \neq 0} + \frac{x}{3} - \frac{\textsf{sign}(x)}{9},\\
	& ~\qquad\text{and}\qquad~ H_n = \sum_{i = 0}^{n} \frac{|x| + \ExpToFun{x \neq 0} \cdot i}{2^{i+1}}
\end{align*}
Next, we need to verify the four preconditions of Theorem~\ref{thm:invariants}.
First, $\charwp{\nicefrac{1}{2}}{C'}{|x|}(G) \leq G$:
\begin{align*}
	\charwp{\nicefrac{1}{2}}{C'}{|x|}(G) &{}~=~ \frac{1}{2}\cdot  |x| + \frac{1}{2} \cdot \wp{C'}{|x| + 1}\\
	&{}~=~ \frac{|x|}{2} + \frac{1}{2} \cdot \Big(\big|-x-\textsf{sign}(x)\big| + 1\Big)\\
	&{}~=~ \frac{|x|}{2} + \frac{| x | + \ExpToFun{x \neq 0} + 1}{2} \\
	&{}~\leq~ \frac{|x|}{2} + \frac{| x | + 1 + 1}{2}\\
	&{}~=~ |x| + 1 ~=~ G
\end{align*}
Second, $\charwp{\nicefrac{1}{2}}{C'}{|x| + x}(I) \leq I$:
\begin{align*}
	&\charwp{\nicefrac{1}{2}}{C'}{|x| + x}(I)\\
	&{}~=~ \frac{1}{2}\cdot  \big(|x| + x\big)\\
	&\qquad{}+ \frac{1}{2} \cdot \wp{C'}{|x| + \ExpToFun{x \neq 0} + \frac{x}{3} - \frac{\textsf{sign}(x)}{9}}\\
	&{}~=~ \frac{|x| + x}{2} + \frac{1}{2} \left(\big|-x-\textsf{sign}(x)\big| + \ExpToFun{-x-\textsf{sign}(x) \neq 0} \vphantom{\frac{(}{3}}\right.\\
	&\qquad\qquad\qquad\quad{} \left.{} + \frac{-x-\textsf{sign}(x)}{3} - \frac{\textsf{sign}(-x-\textsf{sign}(x))}{9}\right)\\
	&{}~=~ \frac{|x| + x}{2} + \frac{|x| + \ExpToFun{x \neq 0} + \ExpToFun{x \neq 0}}{2}\\
	&\qquad{} + \frac{-x-\textsf{sign}(x)}{6} + \frac{\textsf{sign}(x)}{18}\\
	&{}~=~ |x| + \ExpToFun{x \neq 0} + \frac{x}{3} - \frac{\textsf{sign}(x)}{9} ~=~ I ~\leq~ I
\end{align*}
Third, $\charwp{\nicefrac{1}{2}}{C'}{|x|}(\zero) \geq H_0$:
\begin{align*}
	\charwp{\nicefrac{1}{2}}{C'}{|x|}(\zero) &{}~=~ \frac{1}{2}\cdot  |x| + \frac{1}{2} \cdot \wp{C'}{\zero}\\
	&{}~=~ \frac{|x|}{2} + \frac{1}{2} \cdot \zero\\
	&{}~=~ \frac{|x|}{2}\\
	&{}~=~ \frac{|x| + \ExpToFun{x \neq 0} \cdot 0}{2^{0 + 1}}\\
	&{}~=~ \sum_{i = 0}^{0} \frac{|x| + \ExpToFun{x \neq 0} \cdot i}{2^{i + 1}} ~=~ H_0 ~\geq~ H_0
\end{align*}
Fourth, $\charwp{\nicefrac{1}{2}}{C'}{|x|}(H_{n}) \geq H_{n+1}$:
\begin{align*}
	&\charwp{\nicefrac{1}{2}}{C'}{|x|}(H_{n})\\
	&{}= \frac{1}{2}\cdot  |x| + \frac{1}{2} \cdot \wp{C'}{\sum_{i = 0}^{n} \frac{|x| + \ExpToFun{x \neq 0} \cdot i}{2^{i + 1}}}\\
	&{}= \frac{|x|}{2}  + \frac{1}{2} \left(\sum_{i = 0}^{n} \frac{|{-}x-\textsf{sign}(x)| + \ExpToFun{-x-\textsf{sign}(x) \neq 0} \cdot i}{2^{i + 1}}\right)\\
	&{}= \frac{|x| + \ExpToFun{x \neq 0} \cdot 0}{2^{0+1}} + \sum_{i = 0}^{n} \frac{|x| + \ExpToFun{x \neq 0} + \ExpToFun{x \neq 0} \cdot i}{2^{i + 1 + 1}}\\
	&{}= \frac{|x| + \ExpToFun{x \neq 0} \cdot 0}{2^{0+1}} + \sum_{i = 0}^{n} \frac{|x| + \ExpToFun{x \neq 0} \cdot (i + 1)}{2^{i + 1 + 1}}\\
	&{}= \frac{|x| + \ExpToFun{x \neq 0} \cdot 0}{2^{0+1}} + \sum_{i = 1}^{n+1} \frac{|x| + \ExpToFun{x \neq 0} \cdot i}{2^{i + 1}}\\
	&{}= \sum_{i = 0}^{n+1} \frac{|x| + \ExpToFun{x \neq 0} \cdot i}{2^{i + 1}} = H_{n+1} \geq H_{n+1}
\end{align*}
Finally, we have to analyze $\sup_{n \in \mathbb{N}} H_n$, which is fairly straightforward:
\begin{align*}
	\sup_{n \in \mathbb{N}} ~ H_n  & ~=~ \sup_{n \in \mathbb{N}} ~  \sum_{i = 0}^{n} \frac{|x| + \ExpToFun{x \neq 0} \cdot i}{2^{i+1}} \\
	& ~=~ \sum_{i = 0}^{\omega} \frac{|x| + \ExpToFun{x \neq 0} \cdot i}{2^{i+1}} \\
	& ~=~ |x| \cdot \sum_{i = 0}^{\omega} \frac{1}{2^{i+1}} +  \ExpToFun{x \neq 0} \cdot\sum_{i = 0}^{\omega} \frac{i}{2^{i+1}} \\
	& ~=~ |x| \cdot 1 +  \ExpToFun{x \neq 0} \cdot 1 \\
	& ~=~ |x| +  \ExpToFun{x \neq 0}
\end{align*}
Applying \mbox{Theorem}~\ref{thm:invariants} therefore yields
\begin{align*}
	&\wpeqp{\WHILEDO{\nicefrac{1}{2}}{C'}}{x}{|x|}\\
	& ~\sqsubseteq~ \bigEqPair{I - \sup_{n \in \mathbb{N}} H_n}{2 \cdot G}\\
	& ~=~ \bigEqPair{\frac{x}{3} - \frac{\textsf{sign}(x)}{9}}{2 \cdot G}
\end{align*}
Since $2 |x| + 2 < \boldsymbol{\infty}$, we obtain that $\nicefrac{x}{3} - \nicefrac{\textsf{sign}(x)}{9}$ is an over--approximation of the expected value of $x$ after execution of the program for any initial state.

\subsubsection{Analysis Using Jordan Decomposition}

If one wanted to perform an equivalent analysis using the Jordan decomposition of $x$ into $\positive{x} = \max\{x,\, 0\}$ and $\negative{x} = {-}\min\{x,\, 0\}$, one would have to first prove integrability of $\positive{x}$ and $\negative{x}$.
For that, it suffices to find two invariants $G_{\positive{x}}$ and $G_{\negative{x}}$ such that both $\charwp{\nicefrac{1}{2}}{C'}{\positive{x}}\left(G_{\positive{x}}\right) \leq G_{\positive{x}}$ and $\charwp{\nicefrac{1}{2}}{C'}{\negative{x}}\left(G_{\negative{x}}\right) \leq G_{\negative{x}}$.
The simplest invariants we were able to come up with are:
\begin{align*}
	G_{\positive{x}} ~=~ &\ExpToFun{x > 0} \cdot \left( \frac{2}{3} x + \frac{4}{9} \right) + \ExpToFun{x < 0} \cdot \left( {-}\frac{1}{3} x + \frac{5}{9} \right)\\
	G_{\negative{x}} ~=~ &\ExpToFun{x > 0} \cdot \left( \frac{1}{3} x + \frac{5}{9} \right) + \ExpToFun{x < 0} \cdot \left( {-}\frac{2}{3} x + \frac{4}{9} \right)
\end{align*}
Next, we have to establish that those are indeed invariants.
First, $\charwp{\nicefrac{1}{2}}{C'}{\positive{x}}\left(G_{\positive{x}}\right) \leq G_{\positive{x}}$:
\begin{align*}
	&\charwp{\nicefrac{1}{2}}{C'}{\positive{x}}\left(G_{\positive{x}}\right)\\
	&{}= \frac{1}{2}\cdot \positive{x} + \frac{1}{2} \cdot \wpsymbol[{C'}]\left(\ExpToFun{x > 0} \cdot \left( \frac{2}{3} x + \frac{4}{9} \right) \right.\\
	&\qquad\left. {} + \ExpToFun{x < 0} \cdot \left( {-}\frac{1}{3} x + \frac{5}{9} \right) \right)\\
	&{}= \frac{1}{2}\cdot \positive{x}\\
	&\quad~{} + \frac{1}{2} \cdot \left(\ExpToFun{-x-\textsf{sign}(x) > 0} \cdot \left( \frac{2}{3} \big( -x-\textsf{sign}(x) \big) + \frac{4}{9} \right) \right.\\
	&\quad~~\left. {} + \ExpToFun{-x-\textsf{sign}(x) < 0} \cdot \left( {-}\frac{1}{3} \big( -x-\textsf{sign}(x) \big) + \frac{5}{9} \right) \right)\\
	&{}= \frac{1}{2}\ExpToFun{x > 0}x  + \frac{1}{2} \cdot \left(\ExpToFun{x < 0} \cdot \left( \frac{2}{3} \big( -x + 1 \big) + \frac{4}{9} \right) \right.\\
	&\qquad\left. {} + \ExpToFun{x > 0} \cdot \left( {-}\frac{1}{3} \big( -x-1 \big) + \frac{5}{9} \right) \right)\\
	&{}= \ExpToFun{x > 0}\left(\frac{1}{2}x  - \frac{1}{3} \big( -x-1 \big) + \frac{5}{9}\right) \\
	&\qquad {} + \ExpToFun{x < 0} \cdot \left(  \frac{1}{2} \cdot \left(\frac{2}{3} \big( -x + 1 \big) + \frac{4}{9} \right) \right)\\
	&{}= \ExpToFun{x > 0} \cdot \left( \frac{2}{3} x + \frac{4}{9} \right) + \ExpToFun{x < 0} \cdot \left( {-}\frac{1}{3} x + \frac{5}{9} \right)\\
	&{}= G_{\positive{x}} \leq G_{\positive{x}}
\end{align*}
Second, $\charwp{\nicefrac{1}{2}}{C'}{\negative{x}}\left(G_{\negative{x}}\right) \leq G_{\negative{x}}$:
\begin{align*}
	&\charwp{\nicefrac{1}{2}}{C'}{\negative{x}}\left(G_{\negative{x}}\right)\\
	&{}~=~ \frac{1}{2}\cdot \negative{x} + \frac{1}{2} \cdot \wpsymbol[{C'}]\left(\ExpToFun{x > 0} \cdot \left( \frac{1}{3} x + \frac{5}{9} \right) \right.\\
	&\qquad\left. {} + \ExpToFun{x < 0} \cdot \left( {-}\frac{2}{3} x + \frac{4}{9} \right) \right)\\
	&{}= \frac{1}{2}\cdot \negative{x}\\
	&\quad{}\: + \frac{1}{2} \cdot \left(\ExpToFun{-x-\textsf{sign}(x) > 0} \cdot \left( \frac{1}{3} \big( -x-\textsf{sign}(x) \big) + \frac{5}{9} \right) \right.\\
	&\quad~\:\left. {} + \ExpToFun{-x-\textsf{sign}(x) < 0} \cdot \left( {-}\frac{2}{3} \big( -x-\textsf{sign}(x) \big) + \frac{4}{9} \right) \right)\\
	&{}= \frac{1}{2}\ExpToFun{x > 0}x  + \frac{1}{2} \cdot \left(\ExpToFun{x < 0} \cdot \left( \frac{1}{3} \big( -x + 1 \big) + \frac{5}{9} \right) \right.\\
	&\qquad\left. {} + \ExpToFun{x > 0} \cdot \left( {-}\frac{2}{3} \big( -x-1 \big) + \frac{4}{9} \right) \right)\\
	&{}= \ExpToFun{x > 0}\left(\frac{1}{2}x  - \frac{2}{3} \big( -x-1 \big) + \frac{4}{9}\right) \\
	&\qquad {} + \ExpToFun{x < 0} \cdot \left(  \frac{1}{2} \cdot \left(\frac{1}{3} \big( -x + 1 \big) + \frac{5}{9} \right) \right)\\
	&{}= \ExpToFun{x > 0} \cdot \left( \frac{1}{3} x + \frac{5}{9} \right) + \ExpToFun{x < 0} \cdot \left( {-}\frac{2}{3} x + \frac{4}{9} \right)\\
	&{}= G_{\negative{x}} \leq G_{\positive{x}}
\end{align*}
We have by now established integrability of both $\positive{x}$ and $\negative{x}$. 
Furthermore, we know that $G_{\positive{x}}$ is an upper bound of the expected value of $\positive{x}$.

Alternatively, instead of proving integrability of $\positive{x}$ and $\negative{x}$ individually, we could have proved the integrability of $|x|$ using the simpler invariant $G$ from the analysis using integrability witnessing pairs, see Appendix~\ref{app:comparison:iwe}. 
We would then still need to find an upper bound for $\positive{x}$ using invariant $G_{\positive{x}}$, but we would get rid of one of the complicated invariant, namely $G_{\negative{x}}$.

We now need to establish a lower bound for $\negative{x}$.
For that, we need a lower $\omega$--invariant $H_n$.
Again, the simplest we were able to come up with is given by
\begin{align*}
	H_n ~=~ \sum_{i = 0}^{n} {-}\frac{\min\big\{(-1)^i \cdot (x + \textsf{sign}(x) \cdot i ),\, 0\big\}}{2^{i+1}}~.
\end{align*}
For verifying that $H_n$ is indeed an $\omega$--invariant, we need to check two conditions.
First, $\charwp{\nicefrac{1}{2}}{C'}{\negative{x}}(\zero) \geq H_0$:
\begin{align*}
	\charwp{\nicefrac{1}{2}}{C'}{\negative{x}}(\zero) &{}~=~ \frac{1}{2}\cdot \negative{x} + \frac{1}{2} \cdot \wp{C'}{\zero}\\
	&{}~=~ \frac{1}{2}\cdot \negative{x} + \frac{1}{2} \cdot \zero\\
	&{}~=~ \frac{\negative{x}}{2}\\
	&{}~=~ {-}\frac{\min\{x,\, 0\}}{2}\\
	&{}~=~ {-}\frac{\min\big\{(-1)^0 \cdot (x + \textsf{sign}(x) \cdot 0),\, 0\big\}}{2^{0+1}}\\
	&{}~=~ \sum_{i = 0}^{0} {-}\frac{\min\big\{(-1)^i \cdot (x + \textsf{sign}(x) \cdot i ),\, 0\big\}}{2^{i+1}}\\
	&{}~=~ H_0 ~\geq~ H_0
\end{align*}
Second, $\charwp{\nicefrac{1}{2}}{C'}{\negative{x}}(H_{n}) \geq H_{n+1}$:
\begin{align*}
	&\charwp{\nicefrac{1}{2}}{C'}{\negative{x}}(H_{n})\\
	&{}= \frac{1}{2}\cdot  \negative{x}\\
	&\qquad{} + \frac{1}{2} \cdot \wp{C'}{\sum_{i = 0}^{n} {-}\frac{\min\big\{(-1)^i \cdot (x + \textsf{sign}(x) \cdot i ),\, 0\big\}}{2^{i+1}}}\\
	&{}= {-}\frac{\min\{x,\, 0\}}{2} + \frac{1}{2} \sum_{i = 0}^{n} {-}\frac{N}{2^{i+1}}~,\\
	&\text{where } N = \min\big\{(-1)^i \cdot (- x - \textsf{sign}(x)\\
	&\qquad\qquad\qquad\qquad + \textsf{sign}(- x - \textsf{sign}(x)) \cdot i ),\, 0\big\}\\
	&{}= {-}\frac{\min\big\{(-1)^0 \cdot (x + \textsf{sign}(x)\cdot0),\, 0\big\}}{2^{0 + 1}}\\
	&\qquad{} + \frac{1}{2} \sum_{i = 0}^{n} {-}\frac{\min\big\{(-1)^i \cdot (- (x + \textsf{sign}(x)) - \textsf{sign}(x) \cdot i ),\, 0\big\}}{2^{i+1}}\\
	&{}= {-}\frac{\min\big\{(-1)^0 \cdot (x + \textsf{sign}(x)\cdot0),\, 0\big\}}{2^{0 + 1}}\\
	&\qquad{} + \sum_{i = 0}^{n} {-}\frac{\min\big\{(-1)^{i+1} \cdot (x + \textsf{sign}(x) \cdot (i+1) ),\, 0\big\}}{2^{i+1 + 1}}\\
	&{}= {-}\frac{\min\big\{(-1)^0 \cdot (x + \textsf{sign}(x)\cdot0),\, 0\big\}}{2^{0 + 1}}\\
	&\qquad{} + \sum_{i = 1}^{n+1} {-}\frac{\min\big\{(-1)^{i} \cdot (x + \textsf{sign}(x) \cdot i ),\, 0\big\}}{2^{i + 1}}\\
	&{}= \sum_{i = 0}^{n+1} {-}\frac{\min\big\{(-1)^{i} \cdot (x + \textsf{sign}(x) \cdot i ),\, 0\big\}}{2^{i + 1}}\\
	&{}= H_{n+1} \geq H_{n+1}
\end{align*}
We can now argue that an upper bound for the expected value of $x$ is given by
\begin{align*}
	G_{\positive{x}} - \sup_{n \in \mathbb N} H_n~.
\end{align*}
Arguing about $\sup_{n \in \mathbb N} H_n$ is more involved than it was in the case of the integrability--witnessing expectation analysis where obtaining the supremum was immediate:
\begin{align*}
	&\sup_{n \in \mathbb N}~ H_n\\
	& ~=~ \sup_{n \in \mathbb N}~ \sum_{i = 0}^{n} {-}\frac{\min\big\{(-1)^i \cdot (x + \textsf{sign}(x) \cdot i ),\, 0\big\}}{2^{i+1}}\\
	& ~=~ \sum_{i = 0}^{\omega} {-}\frac{\min\big\{(-1)^i \cdot (x + \textsf{sign}(x) \cdot i ),\, 0\big\}}{2^{i+1}}\\
	& ~=~ \sum_{i = 0}^{\omega} {-}\frac{\min\big\{x + \textsf{sign}(x) \cdot 2i,\, 0\big\}}{2^{2i+1}}\\
	&~\qquad + \sum_{i = 0}^{\omega} {-}\frac{\min\big\{ -x - \textsf{sign}(x) \cdot (2i + 1),\, 0\big\}}{2^{2i+2}}\\
	& ~=~ \ExpToFun{x > 0}\left(\sum_{i = 0}^{\omega} {-}\frac{\min\big\{x + \textsf{sign}(x) \cdot 2i,\, 0\big\}}{2^{2i+1}}\right.\\
	&~\qquad\qquad\quad~ \left. + \sum_{i = 0}^{\omega} {-}\frac{\min\big\{ -x - \textsf{sign}(x) \cdot (2i + 1),\, 0\big\}}{2^{2i+2}}\right)\\
	& \qquad {} + \ExpToFun{x < 0}\left(\sum_{i = 0}^{\omega} {-}\frac{\min\big\{x + \textsf{sign}(x) \cdot 2i,\, 0\big\}}{2^{2i+1}}\right.\\
	&~\qquad\qquad\quad~ \left. + \sum_{i = 0}^{\omega} {-}\frac{\min\big\{ -x - \textsf{sign}(x) \cdot (2i + 1),\, 0\big\}}{2^{2i+2}}\right)\\
	& ~=~ \ExpToFun{x > 0}\left(\sum_{i = 0}^{\omega} {-}\frac{0}{2^{2i+1}} + \sum_{i = 0}^{\omega} {-}\frac{-x - (2i + 1)}{2^{2i+2}}\right)\\
	& \qquad {} + \ExpToFun{x < 0}\left(\sum_{i = 0}^{\omega} {-}\frac{x - 2i}{2^{2i+1}} + \sum_{i = 0}^{\omega} {-}\frac{0}{2^{2i+2}}\right)\\
	& ~=~ \ExpToFun{x > 0}\left(\sum_{i = 0}^{\omega} \frac{x + 2i + 1}{2^{2i+2}}\right)  + \ExpToFun{x < 0}\left(\sum_{i = 0}^{\omega} \frac{-x + 2i}{2^{2i+1}}\right)\\
	& ~=~ \ExpToFun{x > 0} \cdot \left( \frac{1}{3} x + \frac{5}{9} \right) + \ExpToFun{x < 0} \cdot \left( {-}\frac{2}{3} x + \frac{4}{9} \right)\tag{obtained with the aid of Wolfram$|$Alpha}\\
	& ~=~ G_{\negative{x}}
\end{align*}
Finally, we can calculate an upper bound for the expected value of $x$ by
\begin{align*}
	&G_{\positive{x}} - \sup_{n \in \mathbb N} H_n\\
	& ~=~ G_{\positive{x}} - G_{\negative{x}}\\
	& ~=~ \ExpToFun{x > 0} \cdot \left( \frac{2}{3} x + \frac{4}{9} \right) + \ExpToFun{x < 0} \cdot \left( {-}\frac{1}{3} x + \frac{5}{9} \right)\\
	&\qquad {}-\left( \ExpToFun{x > 0} \cdot \left( \frac{1}{3} x + \frac{5}{9} \right) + \ExpToFun{x < 0} \cdot \left( {-}\frac{2}{3} x + \frac{4}{9} \right) \right)\\
	& ~=~ \ExpToFun{x > 0} \cdot \left( \frac{2}{3} x + \frac{4}{9} - \frac{1}{3} x - \frac{5}{9}\right) \\
	& \qquad {}+ \ExpToFun{x < 0} \cdot \left( {-}\frac{1}{3} x + \frac{5}{9} + \frac{2}{3} x - \frac{4}{9}\right)\\
	& ~=~ \ExpToFun{x > 0} \cdot \left( \frac{x}{3}- \frac{1}{9}\right) + \ExpToFun{x < 0} \cdot \left(\frac{x}{3}+ \frac{1}{9}\right)\\
	& ~=~ \frac{x}{3} + \frac{\ExpToFun{x > 0} \cdot (-1) + \ExpToFun{x < 0} \cdot 1}{9}\\
	& ~=~ \frac{x}{3} - \frac{\textsf{sign}(x)}{9}
\end{align*}
and we obtain the same result as we did with the integrability--witnessing expectation analysis.

\subsection{Loop Invariants for Lower Bounds}
\label{app:lower-bounds}

\noindent
If we take a closer look at the sequence
\begin{align*}
	\left(\bigEqPair{\charwpn{\pguard}{C'}{|f| + f}{n}(\zero) - \charwpn{\pguard}{C'}{|f|}{n}(\zero)}{\charwpn{\pguard}{C'}{g}{n}(\zero)}\right)_{n\in\mathbb{N}}
\end{align*}
we can see that in order to under--approximate the limit of that sequence, we can---simply put---
\begin{enumerate}
	\item
	over--approximate the limit---i.e.\ the supremum---of $\charwpn{\pguard}{C'}{g}{n}(\zero)$,
	\item
	under--approximate the limit---i.e.\ again the supremum---of $\charwpn{\pguard}{C'}{|f| + f}{n}(\zero)$, and
	\item
	over--approximate the limit---once again: the supremum---of $\charwpn{\pguard}{C'}{|f|}{n}(\zero)$.
\end{enumerate}
Notice that these over- and under--approximations are over- and under--ap{\-}prox{\-}i{\-}ma{\-}tions of standard weakest pre--expectations.
Furthermore, recall that by \mbox{Theorem}~\ref{thm:wp-prop}~(4)~and~(5) we have invariant rules for those over-- and under--approximations.
This immediately leads us to the following proof rule for loops:
\begin{theorem}[Loop Invariants for Lower Bounds of $\boldsymbol{\widetilde{\wpsymbol}}$]
	Let $\EqPair{f}{g} \in \EqPairs$, $C' \in \pProgs$, $I, G \in \E$ with $G(\sigma) < \infty$, for all $\sigma \in \States$, and $(H_n)_{n \in \mathbb N} \subseteq \E$. Then $\charwp{\xi}{C'}{g}(G) \leq G$, $\charwp{\xi}{C'}{|f|}(I) \leq I$, $H_0 \leq \charwp{\pguard}{C'}{|f| + f}(\zero)$, and $H_{n+1} \leq \charwp{\pguard}{C'}{|f| + f}(H_n)$ implies
	\begin{align*}
		\bigEqPair{\sup_{n\in \mathbb N} H_n - I}{2 \cdot G} ~\sqsubseteq~ \wpeqp{\WHILEDO{\xi}{C'}}{f}{g}~.
	\end{align*}
\end{theorem}

\end{document}